\RequirePackage{fix-cm}
\documentclass{svjour3}                     % onecolumn (standard format)
\smartqed  % flush right qed marks, e.g. at end of proof
\usepackage[11pt]{extsizes}

\usepackage{graphicx}
\usepackage{bm}

\usepackage{amssymb,amsmath}
\usepackage{enumerate}
\usepackage{dsfont}
\usepackage[colorlinks, citecolor=red]{hyperref}
\usepackage{comment,cite,color}
\usepackage{cite,color}
\usepackage{mathrsfs}
\usepackage{epsfig}
\usepackage{lscape}
\usepackage{subfigure}
\usepackage{epstopdf}
\usepackage{caption}
\usepackage{algorithm}

\usepackage{tgtermes}
\usepackage{lipsum}

\usepackage{algpseudocode}

 \usepackage{mathptmx}      % use Times fonts if available on your TeX system

\newcommand{\xkh}[1]{\left(#1\right)}
\newcommand{\dkh}[1]{\left\{#1\right\}}
\newcommand{\zkh}[1]{\left(#1\right)}

\newcommand{\nj}[1]{\langle {#1} \rangle}
\newcommand{\innerp}[1]{\langle {#1} \rangle}
\newcommand{\norm}[1]{\|{#1}\|_2}

\newcommand{\norms}[1]{\|{#1}\|}
\newcommand{\abs}[1]{\left\lvert#1\right\rvert}
\newcommand{\Abs}[1]{\lvert#1 \rvert}

\newcommand{\E}{{\mathbb E}}

\newcommand{\PP}{{\mathbb P}}

\newcommand{\1}{{\mathds 1}}

\newcommand{\R}{{\mathbb R}}

\newcommand{\T}{\top}
\newcommand{\C}{{\mathbb C}}

\newcommand{\cA}{{\mathcal A}}

\newcommand{\x}{{\widehat{\bm{x}}}}

\newcommand{\vx}{{\bm x}}
\newcommand{\vw}{{\bm w}}

\newcommand{\vy}{{ \bm{ y}}}

\newcommand{\vu}{{\bm u}}
\newcommand{\vv}{{\bm v}}
\newcommand{\vz}{{\bm z}}
\newcommand{\vb}{{\bm b}}

\newcommand{\ve}{{\bm e}}

\newcommand{\va}{{\bm a}}

\renewcommand\refname{Bibliography}

\newcommand{\cN}{{\mathcal N}}

\renewcommand{\omega}{\eta}

\newcommand{\RNum}[1]{\uppercase\expandafter{\romannumeral #1\relax}}

\renewenvironment{proof}{\noindent {\it \textbf{Proof}~}}{\hfill \qed \par}

\begin{document}

\title{Strong convexity of affine phase retrieval}
%\thanks{Grants or other notes
%about the article that should go on the front page should be
%placed here. General acknowledgments should be placed at the end of the article.}
%\subtitle{Do you have a subtitle?\\ If so, write it here}

%\titlerunning{Short form of title}        % if too long for running head

\author{Meng Huang 
         \and
        Zhiqiang Xu %etc.
}

%\authorrunning{Short form of author list} % if too long for running head

\institute{Meng Huang \at
              School of Mathematical Sciences, Beihang University, Beijing, 100191, China \\
                \email{menghuang@buaa.edu.cn}           %  \\
%             \emph{Present address:} of F. Author  %  if needed
           \and
           Zhiqiang Xu \at
           LSEC, ICMSEC, Academy of Mathematics and Systems Science, Chinese Academy of Sciences, Beijing 100190, China;\\
          School of Mathematical Sciences, University of Chinese Academy of Sciences, Beijing 100049, China\\
              \email{xuzq@lsec.cc.ac.cn}
}

%\date{Received: date / Accepted: date}
% The correct dates will be entered by the editor

\maketitle

\begin{abstract}
The recovery of a signal from the intensity measurements with some entries being known in advance is termed as {\em affine phase retrieval}.
In this paper, we prove  that   a natural least squares formulation for the affine phase retrieval is strongly convex on the entire space under some mild conditions, provided   the measurements  are complex Gaussian random vecotrs  and the measurement number $m \gtrsim d \log d$ where $d$ is the dimension of signals. Based on the result, we prove that  the simple gradient descent method for the affine phase retrieval converges linearly to the target solution with high probability  from an arbitrary initial point. These results show  an essential difference between the  affine  phase retrieval and the classical phase retrieval, where the least squares formulations for the classical phase retrieval are non-convex.

% \PACS{PACS code1 \and PACS code2 \and more}
\end{abstract}

\section{Introduction}
\subsection{Problem setup}
The problem of recovering $\vx$ from the intensity-only measurements
\[
y_j=\abs{\nj{\va_j,\vx}+b_j}^2, \quad j=1,\ldots,m
\]
is termed as {\em affine phase retrieval}. Here,  $\vx\in \C^d$ is an arbitrary unknown vector,
  $\va_j \in \C^d, j=1,\ldots,m$ are known sampling vectors,  $\vb:=(b_1,\ldots,b_m)^\T \in \C^m$ is the bias vector and $y_j \in \R, j=1,\ldots,m$ are observed measurements.
The affine phase retrieval is of significant importance to a number of fields,  such as holography \cite{liebling2003local,latychevskaia,barmherzig,guizar} and Fourier phase retrieval problem \cite{beinert2015,beinert2018,huangK2016,bendory},
where a ``reference'' is situated  or a part of signal is a priori known before capturing the intensity-only measurements.
%Particularly, the affine phase retrieval can be viewed as the classical phase retrieval with some background information.
It has been show theoretically that $m\ge 4d-1$ generic measurements are sufficient to recover all the signals $\vx$ exactly \cite{gaoaffine,huang2021}.

A natural approach to recover the signal $\vx$ is to solve the following program:
\begin{equation} \label{eq:fz}
\min_{\vz\in \C^d} \quad  f(\vz):=\frac1{2m} \sum_{j=1}^m \xkh{\abs{\nj{\va_j,\vz}+b_j}^2-y_j}^2.
\end{equation}
If all $b_j$ are zeros then the above program becomes
\begin{equation}\label{eq:gz}
\min_{\vz\in \C^d} \quad  g(\vz):=\frac1{2m} \sum_{j=1}^m \xkh{\abs{\nj{\va_j,\vz}}^2-y_j}^2,
\end{equation}
which is the intensity-based model \cite{WF,turstregion,TAF,RWF} for solving the classical phase retrieval.
Due to the non-convexity of $g$, the  algorithms for solving \eqref{eq:gz} rely on the carefully-designed initialization heavily \cite{WF,Gaoxu,RWF} or require $g$ possesses the benign geometrical landscape \cite{turstregion,2020a,cai2019}. Our focus is the program \eqref{eq:fz}. So we are interested in the following questions:
{\em Could the program \eqref{eq:fz} be solved  from an arbitrary initial point via the simple gradient descent method? Can we establish the rate of convergence? }

\subsection{Related Work}
\subsubsection{ Phase Retrieval}
The classical phase retrieval problem aims to recover a signal $\vx\in\C^d$ from the intensity-only measurements
\begin{equation} \label{eq;phaR}
y_j=\abs{\nj{\va_j,\vx}}^2, \quad  j=1,\ldots,m.
\end{equation}
It arises in various disciplines and has been investigated recently due to its
 wide range of practical applications in fields of physical sciences and engineering, such as X-ray crystallography \cite{harrison1993phase,millane1990phase}, diffraction imaging \cite{shechtman2015phase,chai2010array}, microscopy \cite{miao2008extending}, astronomy \cite{fienup1987phase}, optics and acoustics \cite{walther1963question,balan2006signal} etc,
where the  detector can record only the diffracted intensity while losing the phase information.  Despite its simple mathematical form, it has been shown that to reconstruct a finite-dimensional discrete signal from its Fourier transform magnitudes is generally NP-complete \cite{Sahinoglou}.

Note that 
$\abs{\nj{\va_j,\vx}}^2=\abs{\nj{\va_j,e^{i\theta}\vx}}^2$ for any $\theta\in \R$. Therefore the recovery of  $\vx$ is up to a global phase for classical phase retrieval.
It was shown  that $m\ge 4d-4$  generic measurements  suffice to recover $\vx$  for the complex case \cite{conca2015algebraic,wangxu} and $m\ge 2d-1$ are sufficient for the real case \cite{balan2006signal}.
In the perspective of algorithms, some efficient gradient descent methods  have  been proposed to solve the classical phase retrieval problem based on  some natural least squares formulations.
Due to the non-convexity of those loss functions, the convergence of the algorithms usually require some sophisticated techniques, such as carefully-designed initialization\cite{huangwang,tan2019phase,WF,TAF,TWF}, benign geometric landscape \cite{turstregion,2020a,cai2019}.
For instance,   in \cite{turstregion}, Sun, Qu and  Wright study the global geometry structure of the following loss function
\[
F(\vz):=\frac{1}{m}\sum_{j=1}^m \xkh{\abs{\nj{\va_j,\vz}}^2-y_j}^2,
\]
and show $F$  does not have any spurious local minima under $m=O(d \log^3 d)$ complex Gaussian random measurements. In other words, all minimizers of $F$ are the target signal $\vx$ up to a global phase,  and there is a negative directional curvature around each saddle point. With this benign geometric landscape in place, the authors of \cite{turstregion} develop a trust-region  method to find a global solution of $\min_\vz F(\vz)$  with random initialization. In fact, armed with these two conditions, the vanilla gradient descent converges almost surely to the global solution with random initialization \cite{Leegradient}, but, to our knowledge,  there is no result about the convergence rate.  To understand the convergence properties of gradient descent with random initialization, Chen et al. \cite{chenrandom} use the ``leave-one-out'' arguments coupled with finer dynamics to prove that the gradient descent with random initialization enjoys nearly linear convergence.
We refer the reader to survey papers  \cite{shechtman2015phase,Chinonconvex,jaganathan2016phase} for accounts of recent developments in the theory, algorithms and applications of phase retrieval.

\subsubsection{Holographic phase retrieval}
The holography was introduced by Gabor in 1948 when he was working on improving the resolution of the invented electron microscope \cite{gabor1948}, and  he was awarded the Nobel Prize in Physics in 1971. In holographic optics,  a reference signal, whose structure is a prior known, is included in the diffraction patterns alongside the signal of interest\cite{latychevskaia,barmherzig,guizar}. Mathematically, when a known reference $\vx' \in \C^k$ is situated to the object $\vx \in \C^d$, it gives $\x:=\left( \begin{array}{l} \vx \\ \vx' \end{array} \right) \in \C^{d+k}$. The  intensity measurements we obtain is
\[
\hat{y}_j=\abs{\nj{\hat{\va}_j,\x}}^2=\abs{\nj{\va_j,\vx}+\nj{\va'_j,\vx'}}^2=\abs{\nj{\va_j,\vx}+b_j}^2, \quad j=1,\ldots,m.
\]
Here, $\hat{\va}_j \in \C^{d+k}$ are the vectors corresponding to the rows of discrete Fourier transform (DFT) matrix, and
we write $\hat{\va}_j:=\left( \begin{array}{l} \va_j \\ \va'_j\end{array} \right)$, $b_j:=\nj{\va'_j,\vx'}$. The recovery of $\vx$ from the measurements $\hat{y}_j$ is the famous {\em holographic phase retrieval} problem, which is an example of the affine phase retrieval.

\subsubsection{The connection  between the affine phase retrieval and the classical phase retrieval}
Recall that the classical phase retrieval aims to recover a signal $\vx\in\C^d$ from the intensity-only measurements
\begin{equation} \label{eq:reclapr}
y_j:=\abs{\nj{\va_j,\vx}}^2, \quad j=1,\ldots,m,
\end{equation}
where $\va_j \in \C^d$  for all $j=1,\ldots,m$. In some practical applications, some entries of $\vx$ might be known in advance, such as the reconstruction of signals in a shift-invariant space from their phaseless samples \cite{chen2020phase}. In such scenarios, if we assume  the first $k$-entries of $\vx$ are known, namely,
\[
\vx:=\left( \begin{array}{l} \vx_1 \\ \vx_2\end{array} \right) \in \C^d
\]
where  $\vx_2 \in \C^{d-k}$ and $\vx_1 \in \C^k$ is a  known vector. If we rewrite
\[
\va_j:=\left( \begin{array}{l} \va_{j,1} \\ \va_{j,2}\end{array} \right) \quad \mbox{with} \quad \va_{j,1}\in \C^k, \va_{j,2} \in \C^{d-k},
\]
then \eqref{eq:reclapr} can be formulated   as
\[
y_j=\abs{\nj{\va_j,\vx}}^2=\abs{\nj{\va_{j,2}, \vx_2}+b_j}^2,
\]
where $b_j:=\va_{j,1}^* \vx_1$ is known.  From the relationship above, we can see that the reconstruction of $\vx_2$ from the intensity-only measurements $y_j$ is exactly the affine phase retrieval problem in $\C^{d-k}$. Thus,  the affine phase retrieval can be viewed as the classical phase retrieval with some background information.

It is well-known that the reconstruction of signals from the intensity of the Fourier transform is  not uniquely solvable \cite{sanz1985,edidin2019}. There exist ambiguities which are  caused by translation, reflection and conjugation, or multiplication with an unimodular constant. These ambiguities are trivial and cannot be avoided. However, besides these trivial ambiguities, there are also $2^{d-2}$ nontrivial ambiguities \cite{beinert2015} for a signal $\vx\in \C^d$. In order to evaluate a meaningful solution of the Fourier phase retrieval, one needs to pose appropriate priori conditions to enforce uniqueness of solutions.  One way to achieve this goal is to use additionally known values of some entries \cite{beinert2018}, which can be recast as affine phase retrieval.

\subsection{Our contributions}
As stated before,  the classical phase retrieval is difficult  to solve due to the non-convexity. It usually requires some sophisticated techniques, such as carefully-designed initialization and benign geometric landscape.  Since the affine phase retrieval
has strong relationship to the classical phase retrieval, so we may ask: {\em Does the algorithms for solving the affine phase retrieval still require such techniques?}

In this paper, we give a negative answer to this problem by showing that
 the loss function $f$ given in \eqref{eq:fz} is strongly convex on the entire space under some mild conditions on $\vb$, and the simple gradient descent method converges linearly to the global solution with an arbitrary initial point,  as stated below.

 \begin{theorem}[Informal]
  Assume that $\vx \in \C^d$ is a fixed vector. Assume that  the vector $\vb\in \C^m$ satisfies $\norm{\vb} \ge  c_0 \sqrt{m}\norm{\vx} $, $ \sum_{j=1}^m |b_j|^4 \lesssim m \norm{\vx}^4 $ and $\norms{\vb}_\infty \lesssim \sqrt{\log m} \norm{\vx}$, where $c_0\geq 3/2$ is a fixed constant. Suppose that $\va_j \in \C^d,j=1,\ldots,m $, are complex Gaussian random vectors with $m\ge C d\log d$.  Then with high probability
 the function $f$ given in \eqref{eq:fz}  is strongly convex on the entire space $\C^d$. Moreover, the Wirtinger flow method with a fixed step size converges linearly to the global solution, from an arbitrary initialization which lies in the  complex ball with radius $R_0:=2\xkh{\frac 1m\sum_{j=1}^m y_j -\norm{\vb}^2/ m }^{1/2}$.  Here, $C>0$ is a universal constant.
\end{theorem}

The theorem asserts that the loss function of the affine phase retrieval has excellent geometric landscape, namely, it possesses the strong convexity property. Thus, solving the affine phase retrieval is as easy as solving a convex problem, which does not need any sophisticated technique  used  in solving the classical phase retrieval.
% Thus, once we know some background information, the classical phase retrieval becomes much easier to solve.

\subsection{Notations}
\subsubsection{Basic notations}
Throughout this paper, we assume the measurements $\va_j\in \C^d, \; j=1,\ldots,m $
are i.i.d. complex Gaussian random vectors and we say  a vector $\va\in \C^d$ is a complex Gaussian
random vector if $\va \sim 1/\sqrt{2}\cdot \cN(0,I_d)+i/\sqrt{2}\cdot \cN(0,I_d)$.
We set $ \mathbb{S}_{\C}^{d-1}:=\{\vz\in \C^d: \norm{\vz}=1\}$.
For a complex number $b$,  we use $b_{\Re}$ and $b_{\Im}$ to denote the real and imaginary part of $b$, respectively.
For any $A,B\in \R$, we use $ A \lesssim B$
to denote $A\le C_0 B$ where $C_0\in \R_+$ is an  absolute constant.  The notion
$\gtrsim$ can be defined similarly.
We use the notations $\norm{\cdot}$ and $\norms{\cdot}_*$ to denote the operator norm and
nuclear norm of a matrix, respectively.
%Moreover, $A \asymp B$ means that there exist constants $C_1,C_2>0$ such that $C_1 A \le B \le C_2 A$.
 Throughout  this paper, $c$, $C$ and the subscript (superscript)   forms of them denote universal constants whose values vary with the
context.

\subsubsection{Wirtinger calculus}
Let $\vz_{\Re} \in \R^d$  and $\vz_{\Im}\in \R^d$ denote the real and imaginary parts of a complex vector $\vz\in \C^d$, respectively.
Consider a real-valued function $f: \C^d\to \R$. According to the  Cauchy-Riemann conditions, $f$ is not complex differentiable unless it is constant.
However, if we view $f(\vz)$ as a function in $(\vx,\vy)\in \R^{d}\times \R^d \cong \C^d$ where $\vx:=\vz_{\Re}, \vy:=\vz_{\Im}$, it is possible that $f(\vx,\vy)$ is
 differentiable in the real sense.  Taking derivative for $f$ with respect to $\vx$ and $\vy$ directly tends to be complicated and tedious. A  simpler way is to adopt  the Wirtinger calculus,  which makes the expressions for derivatives become significantly simpler and resemble those  with respect to $\vx$ and $\vy$  directly.  Here we only present a simple exposition of Wirtinger calculus (see also \cite{WF,Kreutz}).

For any real-valued function $f(\vz)$, we can write it in the form of $f(\vz,\bar{\vz})$, where $\vz=\vx+i\vy$ and $\bar{\vz}:=\vx-i\vy$. Here $\vx:=\vz_{\Re}$ and $\vy:=\vz_{\Im}$.
 If $f$ is differentiable as a function of $(\vx,\vy)\in \R^d\times \R^d$ then the Wirtinger gradient  is well-defined and can be denoted by
 \[
\nabla f(\vz)=\zkh{\frac{\partial f}{\partial \vz}, \frac{\partial f}{\partial \bar{\vz}}}^*,
\]
where
\[
\frac{\partial f}{\partial \vz}:=\frac{\partial f(\vz,\bar{\vz})}{\partial \vz} \Bigg|_{\bar{\vz}=\mbox{constant}}=\zkh{\frac{\partial f(\vz,\bar{\vz})}{\partial z_1 } ,\ldots, \frac{\partial f(\vz,\bar{\vz})}{\partial z_d }  }\Bigg|_{\bar{\vz}=\mbox{constant}}
\]
and
\[
\frac{\partial f}{\partial \bar{\vz}}:=\frac{\partial f(\vz,\bar{\vz})}{\partial \bar{\vz}} \Bigg|_{\vz=\mbox{constant}}=\zkh{\frac{\partial f(\vz,\bar{\vz})}{\partial \bar{z}_1 } ,\ldots, \frac{\partial f(\vz,\bar{\vz})}{\partial \bar{z}_d }  }\Bigg|_{\vz=\mbox{constant}}.
\]
Here, when applying the operator $\frac{\partial f
}{\partial \vz}$, $\bar{\vz}$ is formally treated as a constant, and similar to the operator $\frac{\partial f}{\partial \bar{\vz}}$.
%Note that the partial derivatives $\frac{\partial f}{\partial \vz}$ and $\frac{\partial f}{\partial \bar{\vz}}$ are row vectors.  The Wirtinger gradient is then given by
It has been proved in \cite{remmert1991,brandwood1983} that the partial derivatives $\frac{\partial f}{\partial \vz}$ and $\frac{\partial f}{\partial \bar{\vz}}$ can be equivalently written as
\begin{equation} \label{re:wirstand}
\frac{\partial f}{\partial \vz} =\frac12 \xkh{ \frac{\partial f}{\partial \vx}- i \frac{\partial f}{\partial \vy} } \quad \mbox{and} \quad \frac{\partial f}{\partial \bar{\vz}}=\frac12 \xkh{ \frac{\partial f}{\partial \vx}+ i \frac{\partial f}{\partial \vy} },
\end{equation}
where the partial derivatives with respect to $\vx$ and $\vy$ are standard partial derivatives of the function $f(\vx,\vy):=f(\vz)$ in the real sense.
The Hessian matrix in Wirtinger calculus  is defined as
\[
\nabla^2 f(\vz):=\left[ \begin{array}{ll}
                           \frac{\partial }{\partial \vz} \xkh{\frac{\partial f}{\partial \vz} }^*   & \quad   \frac{\partial }{\partial \bar{\vz}} \xkh{\frac{\partial f}{\partial \vz} }^* \vspace{1em}\\
                            \frac{\partial }{\partial \vz} \xkh{\frac{\partial f}{\partial \bar{\vz}} }^* & \quad  \frac{\partial }{\partial \bar{\vz}} \xkh{\frac{\partial f}{\partial \bar{\vz}} }^*
                            \end{array} \right ].
\]
With the gradient and Hessian in place, Taylor's approximation for $f$ near the point $\vz$ is
\[
f(\vz+ \Delta \vz) \approx f(\vz)+ \xkh{\nabla f(\vz)}^* \left [ \begin{array}{l} \Delta \vz  \vspace{0.5em}\\ \overline{\Delta \vz } \end{array}\right] +\frac12  \left [ \begin{array}{l} \Delta \vz \vspace{0.5em} \\ \overline{\Delta \vz } \end{array}\right]^* \nabla^2 f(\vz)  \left [ \begin{array}{l} \Delta \vz \vspace{0.5em}\\ \overline{\Delta \vz } \end{array}\right]
\]
for a small perturbation $\Delta \vz \in \C^d$.
For a real-valued function $f$, $\vz$ is a  stationary point  if and only if the Wirtinger gradient obeys
\[
\nabla f(\vz) = \bm{0}.
\]
The curvature of $f$ at a stationary point $\vz$ is dictated by the Wirtinger Hessian $\nabla^2 f(\vz)$. An important observation is that the Hessian quadratic form involves left and right multiplication with a $2d$-dimensional vector consisting of a conjugate pair $\xkh{\Delta \vz, \overline{\Delta \vz} }$. This gives the definition of strongly convex for a real-valued function $f$.

\begin{definition} \label{def:strongconv}
A real-valued function $f: \C^d\to \R$ is called strongly convex on the entire space with a constant $c_0>0$ if
\[
  \left( \begin{array}{l}
\vv\\ \bar{\vv} \end{array} \right)^* \nabla^2 f(\vz)  \left( \begin{array}{l}
\vv\\ \bar{\vv} \end{array} \right)\,\, \ge\,\, c_0 \norm{\vv}^2,\quad \text{ for all }\vz, \vv \in \C^{d}.
\]

\end{definition}

\begin{remark} \label{Re:relastrongcon}
For a differentiable function $h: \R^d \to \R$, the standard definition of strongly convex with a parameter $\beta>0$ is
\[
h(\vu) \ge h(\vw)+\nabla h(\vw)^\T (\vu-\vw) +\frac {\beta}2 \norm{\vu-\vw}^2 \quad \mbox{for all} \quad \vu,\vw \in \R^d.
\]
Here, the $\nabla h$ is the standard gradient of the function $h$. In fact,  Definition \ref{def:strongconv} is equivalent to the above standard definition of strong convexity.
To show it, we observe that if
\begin{equation}\label{eq:c0}
  \left( \begin{array}{l}
\vv\\ \bar{\vv} \end{array} \right)^* \nabla^2 f(\vz)  \left( \begin{array}{l}
\vv\\ \bar{\vv} \end{array} \right) \ge c_0 \norm{\vv}^2 \quad \mbox{for all} \quad \vv\in\C^d,
\end{equation}
it then follows from  \eqref{re:wirstand} and the fundamental theorem of calculus that
\begin{eqnarray}
f(\vz+\Delta \vz) &=& f(\vz)+\int_0^1  \xkh{\nabla f(\vz+ t \Delta \vz)}^* \left [ \begin{array}{l} \Delta \vz  \vspace{0.5em}\\ \overline{\Delta \vz } \end{array}\right] dt \nonumber \\
&=&  f(\vz)+ \xkh{\nabla f(\vz)}^* \left [ \begin{array}{l} \Delta \vz  \vspace{0.5em}\\ \overline{\Delta \vz } \end{array}\right] + \int_0^1\int_0^t  \left [ \begin{array}{l} \Delta \vz \vspace{0.5em} \nonumber \\ \overline{\Delta \vz } \end{array}\right]^* \nabla^2 f(\vz+ \tau \Delta \vz)  \left [ \begin{array}{l} \Delta \vz \vspace{0.5em}\\ \overline{\Delta \vz } \end{array}\right] d\tau dt\\
&\ge & f(\vx,\vy) + \left[\nabla_{\vx} f(\vx,\vy)^\T  \quad \nabla_{\vy} f(\vx,\vy)^\T\right] \left [ \begin{array}{l} \Delta \vx  \vspace{0.5em}\\ \Delta \vy \end{array}\right]+ \frac{c_0}2  \left\| \left [ \begin{array}{l} \Delta \vx  \vspace{0.5em}\\ \Delta \vy \end{array}\right]\right\|^2 \label{eq:strrele}
\end{eqnarray}
for all $\vz,\Delta \vz \in \C^d$.
Here,  $\vx:=\vz_{\Re}, \vy:=\vz_{\Im}, \Delta \vx:= (\Delta \vz)_{\Re}$  and $\Delta \vy:= (\Delta \vz)_{\Im}$. We view $f(\vx,\vy):=f(\vz)$ as a function of  $(\vx,\vy)\in \R^d\times \R^d$. We use $\nabla_{\vx} f(\vx,\vy), \nabla_{\vy} f(\vx,\vy)$  to denote the standard gradients of $f(\vx,\vy)$ with respect to $\vx,\vy$.
Since $f(\vz+\Delta \vz)=f(\vx+\Delta \vx, \vy+\Delta \vy)$, it then follows from \eqref{eq:strrele} that $f(\vz)=f(\vx,\vy)$ is strongly convex with parameter $c_0>0$. Here, $c_0$ is defined in (\ref{eq:c0}).
\end{remark}

%Strongly convex functions have useful properties in optimization theory.
%For instance, if $f$ is strongly convex, then it is bounded from below and has a unique minimum on every compact set.

Strong convexity is one of the most important concepts in optimization, especially for guaranteeing linearly convergence of many gradient descent based algorithms.
For the loss function $f$ defined in \eqref{eq:fz}, direct calculation gives the Wirtinger gradient
\begin{equation} \label{eq:wirtgrad}
\nabla f(\vz)=\frac1m \sum_{j=1}^m \left[ \begin{array}{l}
 \xkh{|\va_j^* \vz+b_j|^2-y_j}  \xkh{\va_j^* \vz+b_j} \va_j \vspace{1em}\\
 \xkh{ |\va_j^* \vz+b_j|^2-y_j } \xkh{\va_j^\T \bar{\vz}+\bar{b}_j} \bar{\va}_j
  \end{array} \right]
\end{equation}
and the Hessian matrix
\begin{equation} \label{eq:Hessian}
\nabla^2 f(\vz)=\frac1m \sum_{j=1}^m \left[ \begin{array}{cc}
\xkh{2 | \va_j^* \vz+b_j|^2-y_j} \va_j \va_j^* & ( \va_j^* \vz+b_j)^2\va_j\va_j^\T \vspace{1em}\\
( \vz^* \va_j +\bar{b}_j)^2\bar{\va}_j\va_j^* & \xkh{2 | \va_j^* \vz+b_j|^2-y_j} \bar{\va}_j \va_j^\T
\end{array} \right].
\end{equation}

\subsection{Organization}
The paper is organized as follows.  In Section 2, we demonstrate that the natural least squares formulation \eqref{eq:fz} for the affine phase retrieval is strongly convex, which means the loss function exhibits the excellent global landscape. Based on this characterization, in Section 3 we show that the Wirtinger flow
 for solving the affine phase retrieval from an arbitrary initial point is linearly convergent. In Section 4, we study the empirical performance of our algorithm via a series of numerical experiments. In Section 5, we present a brief discussion for the future work.  Appendixes A and B collect the technical lemmas needed in our analysis and the detailed proofs to technical results, respectively.

\section{The strong convexity of the objective function}
In this section we demonstrate that the objective function $f$ given in \eqref{eq:fz}  is strongly convex on the entire space.
The intuition is as follows:
since $\va_j$ are complex Gaussian random vectors, it is easy to check the Hessian matrix \eqref{eq:Hessian}  is  strongly convex in expectation,  under some suitable conditions on $\vb$.
%On the other hand, because  the loss function \eqref{eq:fz} is a sum of independent random variables,  the Hessian matrix should be uniformly close to its expectation when the number of measurements $m$ is sufficiently large.
However, the loss function  \eqref{eq:fz} is heavy-tailed because it involves the third powers and the fourth powers of Gaussian random variables. Thus, to ensure the Hessian matrix is uniformly close to its expectation directly, it requires $m\ge C d^2$ samples. This is a sub-optimal result since $m=O(d)$ measurements suffice to guarantee the uniqueness of the affine phase retrieval.

To address this issue, we truncate the terms which involve the third powers or the fourth powers of Gaussian random variables into two parts. For the first part,  it is well-behaved; and for the second part, it is heavy-tailed but can be bounded by another {\em nonnegative} term which
in the form of fourth powers of Gaussian random variables. Because it is nonnegative, its deviation below its expectation is bounded, which means the lower tail is well-behaved. By exploiting this technique, we can prove the main result:

\begin{theorem} \label{th:formmain}
Assume that $\vx \in \C^d$ is an arbitrary  fixed vector and the vector $\vb\in \C^m$ satisfies $ \norm{\vb}  \ge c_0 \sqrt{m}\norm{\vx}  $,  $ \sum_{j=1}^m |b_j|^4 \lesssim m \norm{\vx}^4 $ and $\norms{\vb}_\infty \lesssim \sqrt{\log m} \norm{\vx}$, where $c_0>\sqrt{4.4/1.96}$ is a positive constant.  Suppose that $\va_j \in \C^d,j=1,\ldots,m $, are complex Gaussian random vectors and $y_j=\abs{\innerp{\va_j,\vx}+b_j}^2, j=1,\ldots,m $. If $m\ge C d\log d$  then with probability at least $1-c_a m^{-1}-18 \exp(-c_d d)-c_b \exp(-c_c  m /\log m)$ the Hessian matrix of  $f$ given in \eqref{eq:Hessian} obeys
 \[
  \left( \begin{array}{l}
\vv\\ \bar{\vv} \end{array} \right)^* \nabla^2 f(\vz)  \left( \begin{array}{l}
\vv\\ \bar{\vv} \end{array} \right) \ge \xkh{1.96c_0^2-4.4} \norm{\vx}^2 \quad \mbox{for all} \quad \vv \in \mathbb{S}_\C^{d-1}, \vz \in \C^d.
\]
Here, $C, c_a, c_b, c_c$ and $c_d$ are positive universal constants.
\end{theorem}
%\begin{remark}
%Note that if $c_0\ge 1.5$ then $1.96c_0^2-4.4>0$.
From Theorem \ref{th:formmain} and Definition  \ref{def:strongconv}, we immediately obtain that with high probability  the loss function $f$ defined in  \eqref{eq:fz}  is strongly convex for all $\vz \in \C^d$ under some suitable condition on $\vb$.
%\end{remark}

\begin{remark}
In Theorem \ref{th:formmain}, we require that the bias vector $\vb\in \C^m$ satisfies  the conditions $ \norm{\vb}  \ge c_0 \sqrt{m}\norm{\vx}  $,  $ \sum_{j=1}^m |b_j|^4 \lesssim m \norm{\vx}^4 $ and $\norms{\vb}_\infty \lesssim \sqrt{\log m} \norm{\vx}$ for some fixed constant $c_0>0$. In fact, there exist many vectors satisfying them. For instance, if each entry of $\vb$ is generated independently according to the Gaussian distribution, i.e., $b_j \sim \lambda\norm{\vx}\cdot \mathcal{N}(\mu,\sigma^2 )$ where $\mu, \sigma$ are arbitrary  constants with $\sigma\neq 0$, then the vector $\vb$ satisfies those conditions with high probability provided the parameter $\lambda\ge k c_0/\sigma$ for a universal constant $k>0$.
\end{remark}

{\noindent\it \textbf{Proof of Theorem} \ref{th:formmain}}~
Without loss of generality, we assume that $\norm{\vx}=1$ (the general case can  be obtained via a simple rescaling ).
For any unit vector $\vv\in \C^d$, let
\[
\vu:=\left( \begin{array}{l}
\vv\\ \bar{\vv} \end{array} \right) \in \C^{2d}.
\]
%Through some algebraic computation, it is easy to check
It then follows from (\ref{eq:Hessian}) that
\begin{equation}\label{eq:mainL}
\begin{aligned}
 &\vu^* \nabla^2 f(\vz) \vu \\
 &\quad= \frac2m \sum_{j=1}^m \xkh{2 | \va_j^* \vz+b_j|^2-| \va_j^* \vx+b_j|^2}  |\va_j^* \vv|^2+ \frac2m \sum_{j=1}^m \zkh{( \va_j^* \vz+b_j)^2(\vv^* \va_j)^2}_{\Re} \\
 &\quad \ge  \frac2m \sum_{j=1}^m | \va_j^* \vz+b_j|^2 |\va_j^* \vv|^2- \frac2m \sum_{j=1}^m | \va_j^* \vx+b_j|^2 |\va_j^* \vv|^2 + \frac4m \sum_{j=1}^m \zkh{( \va_j^* \vz+b_j)(\vv^* \va_j)}_{\Re}^2   \\
  &\quad\ge    \frac2m \sum_{j=1}^m \xkh{|\va_j^* \vz|^2 |\va_j^* \bm{v}|^2- |\va_j^* \vx|^2|\va_j^* \vv|^2+2 \zkh{ (\va_j^* \vz)(\vv^* \va_j)}_{\Re}^2+ 2 \zkh{\bar{b}_j (\va_j^* \vz)}_{\Re} |\va_j^* \vv|^2}\\
  & \qquad + \frac2m \sum_{j=1}^m \xkh{2 \zkh{\bar{b}_j(\va_j^* \vv) }_{\Re}^2 -  2\zkh{\bar{b}_j (\va_j^* \vx)}_{\Re} |\va_j^* \vv|^2-4 \zkh{\bar{b}_j(\va_j^* \vv) )}_{\Re} \zkh{(\va_j^* \vz) (\va_j^* \vv)}_{\Re}}.
  \end{aligned}
 \end{equation}
For convenience, we set
\[
\cA(\vz, \vv)\,\,:=\,\, \frac{1}{m} \sum_{j=1}^m \abs{\va_j^* \vz}^2 |\va_j^* \vv|^2.
\]
 We claim that  if $m\ge c(\epsilon) d\log d$ then with probability at least $1-c_1 \exp(-c_2(\epsilon) m/\log m ) - c_3(\epsilon) m^{-1}- 18\exp(-c_4 d)$, it holds that
 \begin{equation} \label{eq:loudaf}
 \begin{aligned}
 \frac12 \vu^* \nabla^2 f(\vz) \vu &\ge \cA(\vz,\vv) -6\epsilon \xkh{ \frac{\norm{\vb}}{\sqrt m}+1} \xkh{\cA(\vz,\vv)}^{\frac 12} +  (1-\epsilon) \cdot \frac{\norm{\vb}^2}{m} \\
&\quad -(1+\epsilon)  -|\vx^* \vv|^2-6\epsilon\xkh{ \frac{\norm{\vb}}{\sqrt m}+1} \norm{\vz}-2\epsilon\xkh{ \frac{\norm{\vb}}{\sqrt m}+1}-2\epsilon.
 \end{aligned}
 \end{equation}
 Here, $\epsilon$ is any constant in $(0,1)$, $c_1, c_4$ are positive universal constants,  $c(\epsilon), c_2(\epsilon)$ and $c_3(\epsilon)$ are positive  constants depending only on $\epsilon$.
 Set
 \[
 \mathcal{R}:=\dkh{(\vz,\vv)\in \C^d\times \mathbb{S}_{\C}^{d-1}: \cA(\vz,\vv)\ge 1}.
 \]
  To give a lower bound for $  \vu^* \nabla^2 f(\vz) \vu$, we divide the space $\C^d\times \mathbb{S}_{\C}^{d-1}$ into two regimes: $(\vz,\vv) \in  \mathcal{R}$ and $(\vz,\vv) \notin  \mathcal{R}$.

  {\bf Regime 1:}  If $(\vz,\vv) \in  \mathcal{R}$ then  we have
  \begin{equation} \label{eq:reg1less}
  \xkh{\cA(\vz,\vv)}^{\frac 12}  \le \cA(\vz,\vv).
  \end{equation}
By  Lemma \ref{le:sunju22},  we obtain that when $m\ge c(\epsilon) d \log d$, with probability at least $1-c'  m^{-d}-c''\exp(-c'''(\epsilon) m )$, it holds that
\begin{equation} \label{eq:reg1great}
\cA(\vz,\vv) \ge (1-\epsilon) \xkh{\norm{\vz}^2+|\vz^* \vv|^2}  \quad \mbox{for all} \quad \vz \in \C^d, \vv\in \mathbb{S}_{\C}^{d-1}.
\end{equation}
Note that  $\norm{\vb} \ge c_0 \sqrt m$ and $\norm{\vb} \le \sqrt[4]{m \sum_{j=1}^m |b_j|^4} \le C_1\sqrt m$ for a universal constant $C_1\ge 1$. Putting \eqref{eq:reg1less} and \eqref{eq:reg1great} into \eqref{eq:loudaf} and taking $\epsilon:=\frac 1{60(C_1+1)}$,  we obtain that,   when $m\ge c(\epsilon) d \log d$, with probability at least $1-c_a m^{-1}-c_b \exp(-c_c  m /\log m)-18 \exp(-c_d d)$, it holds that
   \begin{eqnarray*}
 \frac12 \vu^* \nabla^2 f(\vz) \vu &\ge&\xkh{1 -6\epsilon \xkh{ \frac{\norm{\vb}}{\sqrt m}+1} } (1-\epsilon)\xkh{\norm{\vz}^2+ |\vz^* \vv|^2}  +  (1-\epsilon) \cdot \frac{\norm{\vb}^2}{m}\\
 &&-(1+\epsilon)  -|\vx^* \vv|^2-6\epsilon\xkh{ \frac{\norm{\vb}}{\sqrt m}+1} \norm{\vz}-2\epsilon\xkh{ \frac{\norm{\vb}}{\sqrt m}+1} -2\epsilon  \\
 &\ge &0.88\norm{\vz}^2-0.1\norm{\vz} +0.98 \cdot  \frac{\norm{\vb}^2}{m}-2.1      \\
  &\ge & 0.98c_0^2-2.2,
 \end{eqnarray*}
  where we use the fact that $ 0.88 t^2 -0.1t +0.1>0$ for any $t\ge 0$ in the last inequality. Here, $c_a, c_b, c_c$ and $c_d$ are positive universal constants.

  {\bf Regime 2:  }   If $(\vz,\vv) \notin  \mathcal{R}$ then we have
  \[
\cA(\vz,\vv)\,\,<\,\,1.
  \]
  Similarly, taking $\epsilon:=\frac 1{60(C_1+1)}$ in (\ref{eq:loudaf}), we obtain that
   \begin{eqnarray*}
 \frac12 \vu^* \nabla^2 f(\vz) \vu
 &\ge& \cA(\vz,\vv) -6\epsilon \xkh{ \frac{\norm{\vb}}{\sqrt m}+1} +  (1-\epsilon) \cdot \frac{\norm{\vb}^2}{m} \\
 &&-(1+\epsilon)  -|\vx^* \vv|^2-6\epsilon\xkh{ \frac{\norm{\vb}}{\sqrt m}+1} \norm{\vz}-2\epsilon\xkh{ \frac{\norm{\vb}}{\sqrt m}+1}-2\epsilon. \\
  &\ge& (1-\epsilon)\xkh{\norm{\vz}^2+ |\vz^* \vv|^2}  +  (1-\epsilon) \cdot \frac{\norm{\vb}^2}{m} -6\epsilon\xkh{ \frac{\norm{\vb}}{\sqrt m}+1} \norm{\vz} \\
 && -8\epsilon\xkh{ \frac{\norm{\vb}}{\sqrt m}+1} -2 -3\epsilon  \\
  &\ge & 0.98c_0^2-2.2.
 \end{eqnarray*}
 Here, the second inequality follows from  (\ref{eq:reg1great}) and the fact of $ |\vx^* \vv| \le 1$ .
 Combining the results above, we arrive at the conclusion.

 It remain to prove  (\ref{eq:loudaf}). Our main idea is to bound the terms in (\ref{eq:mainL}).
According to Lemma \ref{le:sunju},  when $m\ge c(\epsilon) d \log d$,  with probability at least
 $1-c_5 \epsilon^{-2} m^{-1}-c_6 \exp(-c_7 \epsilon^2  m /\log m)$, it holds that
\begin{eqnarray*}
\frac1m \sum_{j=1}^m |\va_j^* \vx|^2|\va_j^* \vv|^2 &=&\vv^*\xkh{ \frac1m \sum_{j=1}^m |\va_j^* \vx|^2\va_j \va_j^*} \vv \\
 & \le & 1+\epsilon+ |\vx^* \vv|^2  \qquad \mbox{for all} \quad \vv \in \mathbb{S}_{\C}^{d-1},
\end{eqnarray*}
where we use the fact that $\norm{\vx}=1$ in the inequality. Here, $\epsilon$ is any constant in $(0,1)$, $c_5, c_6$ and $c_7$ are positive universal constants.
From Lemma \ref{le:sunju22}, we obtain that  the following holds with probability at least $1-c'  m^{-d}-c''\exp(-c'''(\epsilon) m )$
\begin{eqnarray*}
\frac{1}{m} \sum_{j=1}^m \zkh{(\va_j^* \vz)(\vv^* \va_j)}_{\Re}^2 &\ge&  \frac{1-\epsilon}2 \xkh{\norm{\vz}^2+ 3(\vz^* \vv)_{\Re} - (\vz^* \vv)_{\Im} } \\
 &\ge & \frac{1-\epsilon}2\xkh{\norm{\vz}^2 - |\vz^* \vv|^2} \quad \mbox{for all} \quad \vz \in \C^d, \vv\in \mathbb{S}_{\C}^{d-1},
\end{eqnarray*}
provided   $m\ge c(\epsilon) d \log d$,
where $c',c''$ are positive universal constants and  $ c'''(\epsilon)>0$ is a constant depending only on $\epsilon$. Recall that $\norm{\vx}=1$, $ \sum_{j=1}^m  |b_j|^4 \lesssim m$ and $\norms{\vb}_{\infty} \le \sqrt{\log m}$.
It follows from Lemma \ref{le:azav2} that,  for $m\ge c(\epsilon) d \log m$,  with probability at least $1- 6\exp(-c_4 d) -6 \exp(-c'''(\epsilon)~ m/\log m )  $,   it holds that
\[
\abs{\frac 1m \sum_{j=1}^m  \zkh{\bar{b}_j (\va_j^* \vz)}_{\Re} |\va_j^* \vv|^2 }\le  \epsilon\xkh{ \frac{\norm{\vb}}{\sqrt m}+1} \xkh{\norm{\vz} + \xkh{\cA(\vz,\vv)}^{\frac 12}  }
\]
for all $\vz\in \C^d, \vv \in \mathbb{S}_\C^{d-1}$. Here, $c_4>0$ is a universal constant.
 Applying Lemma \ref{le:rebax}, the following holds with probability at least
$1-6 \exp(-c_7 \epsilon^2 m /\log m) - 6\exp(-c_4 d)  $:
\[
\frac 1m \sum_{j=1}^m \zkh{\bar{b}_j (\va_j^* \vv)}_{\Re}^2  \ge \frac{1-\epsilon}{2} \cdot \frac{\norm{\vb}^2}{m} -\epsilon \quad \mbox{for all} \; \vv \in \mathbb{S}_{\C}^{d-1},
\]
provided $m\ge C' \epsilon^{-2}\log (1/\epsilon) ~d \log m$, where $C'>0$ is  a universal constant. Recognize that $\norm{\vb} \lesssim \sqrt{m}$. It can be deduced from  Lemma \ref{le:axav2} that when $m\ge C' \epsilon^{-2} d \log m$,  the following holds with probability at least $1-c_5 \epsilon^{-2} m^{-1}-2 \exp(- c_7 \epsilon^2  m/\log m ) $:
\[
\abs{\frac 1m \sum_{j=1}^m \zkh{\bar{b}_j (\va_j^* \vx)}_{\Re} |\va_j^* \vv|^2} \le  \epsilon\xkh{\frac{ \norm{\vb}}{\sqrt m}+ 1}  \quad \mbox{for all} \; \vv \in \mathbb{S}_\C^{d-1}.
\]
Finally,  Lemma \ref{le:rbavrazav} implies that when $m\ge c(\epsilon) d \log m$  with probability at least $1- 6\exp(-c_4 d) -6 \exp(-c'''(\epsilon) m/\log m )  $  it holds
\[
\abs{\frac 1m \sum_{j=1}^m \zkh{\bar{b}_j(\va_j^* \vv) }_{\Re} \zkh{(\va_j^* \vz) (\va_j^* \vv)}_{\Re} } \le  \epsilon\xkh{ \frac{\norm{\vb}}{\sqrt m}+1} \xkh{\norm{\vz} + \xkh{\cA(\vz,\vv)}^{\frac 12}  }
\]
for all $\vz\in \C^d, \vv \in \mathbb{S}_{\C}^{d-1}$.
Substituting the results above into  (\ref{eq:mainL}), we obtain  (\ref{eq:loudaf}).
\qed

\section{Optimization by Wirtinger Gradient Descent}
\label{se:opGD}
Based on the strongly convex of $f$, we could solve the program \eqref{eq:fz} by the following vanilla Wirtinger gradient descent
\[
\vz_{k+1}=\vz_k-\mu \nabla f(\vz_k)
\]
with an arbitrary initial point.
  Here, with abuse of notation,  we set
\begin{equation} \label{eq:abugrad}
\nabla f(\vz):=\frac1m \sum_{j=1}^m
 \xkh{|\va_j^* \vz+b_j|^2-y_j}  \xkh{\va_j^* \vz+b_j} \va_j\,\, \in\,\, \C^d.
 \end{equation}
 Compared with (\ref{eq:wirtgrad}),  the $\nabla f(\vz)\in \C^{d}$ in (\ref{eq:abugrad}) just keeps the first $d$ entries of (\ref{eq:wirtgrad})
 due to the fact that the second part of \eqref{eq:wirtgrad} is the conjugate of the first.

%For our affine phase retrieval problem, $\norm{\vx}$ is not known in advance. However, it can be well estimated, as shown below.
The following lemma  presents an upper bound for $\norm{\vx}$ which is useful for choosing an initial guess.
\begin{lemma} \label{le:bound4x}
 Assume that $\vx \in \C^d$ is an arbitrary fixed vector and $\vb=(b_1,\ldots,b_m)^\top \in \C^m$ is a vector satisfying $\norm{\vb} \lesssim \sqrt{m} \norm{\vx}$. Suppose   $y_j=|\va_j^* \vx +b_j|^2, j=1,\ldots,m$ where $\va_j \in \C^d$ are i.i.d complex Gaussian random vectors.
 Then, with probability at least $1-4\exp(-c m)$, the following   holds
\[
R_0/3\leq  \norm{\vx}\leq R_0.
\]
Here, $R_0:=2\xkh{\frac 1m\sum_{j=1}^m y_j -\frac{\norm{\vb}^2}{m} }^{1/2}$ and  $c>0$ is a universal constant.
\end{lemma}
\begin{proof}
See Appendix \ref{appendix:B}.
\end{proof}

%Recall that $y_j=|\va_j^* \vx +b_j|^2, j=1,\ldots,m$. According to   Lemma \ref{le:bound4x}, the following holds  with probability at least $1-4\exp(-c_1 m)$:
%\begin{equation} \label{eq:estR0}
%\frac14 \norm{\vx}^2 \le \frac 1m\sum_{j=1}^m y_j -\frac{\norm{\vb}^2}{m} \le \frac74 \norm{\vx}^2,
%\end{equation}
%where $c_1>0$ is a universal constant, which implies
%\[
%\norm{\vx}\leq  R_0:=2\xkh{\frac 1m\sum_{j=1}^m y_j -\frac{\norm{\vb}^2}{m} }^{1/2},
%\]
% holds with high probability.
Based on Lemma \ref{le:bound4x}, we can  choose an  initial point $\vz_0$ over $\mathbb{B}_{\C}^d(R_0):=\dkh{\vz\in \C^d: \norm{\vz} \le R_0}$ {\em arbitrarily}.  This gives the following algorithm:

\begin{algorithm}[H]
\caption{Gradient Descent Algorithm for Affine Phase Retrieval with Arbitrary Initial Point}
\label{Al:1}
\begin{algorithmic}[H]
\Require
Measurement vectors: $\va_j \in \C^d, j=1,\ldots,m $; Bias vector: $\vb \in \C^m$;
Observations: $\vy \in \R^m$;  Step size $\mu$; The maximum number of iterations $T$.  \\
\begin{enumerate}
\item[1:] Choose  $\vz_0\in \mathbb{B}_{\C}^d(R_0)$ as an initial guess.
\item[2:] {\bf Loop:}

 {\bf for}  $k=0$ {\bf to} $T-1$ {\bf do}
\[
\vz_{k+1}=\vz_{k}-\mu \nabla f(\vz_{k})
\]
\item[3:] {\bf end for}
\end{enumerate}

\Ensure
The vector $ \vz_T$.
\end{algorithmic}
\end{algorithm}

Next, we prove the Algorithm \ref{Al:1} converges to the target solution $\vx$ linearly. To this end, we need to  provide the Lipschitz constant of the Wirtinger derivative $\nabla f(\vz)$, as shown below.
%To this end, we need to consider our algorithm in a bound regime.
%\subsection{Local smoothness property}
\begin{lemma}[Local Smoothness Property] \label{le:Lipcondition}
Suppose that $\va_j \in \C^d, j=1,\ldots,m,$ are i.i.d complex Gaussian random vectors. Let $\mathcal{S}_R:=\dkh{\vz\in \C^d: \norm{\vz}\le R}$ be a bounded region where $R$ is any  positive constant. If $m\ge Cd$  then with probability at least $1-4\exp(-cm)-c_a m^{-d}$, the Wirtinger gradient $\nabla f(\vz)$ given in \eqref{eq:wirtgrad} is Lipschitz continuous over $\mathcal{S}_R$, i.e.,
\[
\norm{\nabla f(\vz) -\nabla f(\vz')} \le C_R \norm{\vz-\vz'} \quad \mbox{for all} \quad \vz,\vz' \in \mathcal{S}_R,
\]
where
\[
C_R=6\sqrt 2 \xkh{2R d\log m +\norms{\vb}_{\infty} \sqrt{d\log m} }\xkh{R+\frac{ \norm{\vb}}{\sqrt m}}+8\sqrt 2 \Big( 2d\log m (R^2+\norm{\vx}^2)+ \norms{\vb}_{\infty}^2 \Big).
\]
Here, $\vb\in \C^m$ is an arbitrary vector, $C, c$ and $c_a$ are positive universal constants.
\end{lemma}
\begin{proof}
See Appendix \ref{appendix:B}.
\end{proof}

Based on strongly convex and local smoothness properties as stated in Theorem \ref{th:formmain} and Lemma \ref{le:Lipcondition} respectively, we are ready to present the convergence property of Algorithm  \ref{Al:1}.

\begin{theorem}
Assume that $\vx \in \C^d$ is an arbitrary  fixed vector.
Assume that  the vector $\vb\in \C^m$ satisfies $   \norm{\vb}  \ge c_0 \sqrt{m}\norm{\vx}$,   $ \sum_{j=1}^m |b_j|^4 \le c_1 m \norm{\vx}^4 $ and $\norms{\vb}_\infty \le c_2 \sqrt{\log m} \norm{\vx}$, where $c_0>{\sqrt{4.4/1.96}}$ and $c_1, c_2>0$ are positive constants.  Suppose that $\va_j \in \C^d,j=1,\ldots,m $, are complex Gaussian random vectors and $y_j=\abs{\innerp{\va_j,\vx}+b_j}^2, j=1,\ldots,m $. If $m\ge C d\log d$  then with probability at least $1-c_a m^{-1}-c_b \exp(-c_c  m /\log m)-26 \exp(-c_d d)$, the iteration $\vz_k$ given by Algorithm \ref{Al:1} with a fixed step size $\mu \le c_3/(d\log m \norm{\vx}^2)$ obeys
\[
\norm{\vz_{k}-\vx}^2 \le 16 (1-\rho)^k \norm{\vx}^2,
\]
where $\rho:=\mu(1.96c_0^2-4.4) <1$.
Here, $C, c_a, c_b, c_c, c_d$ are positive universal constants and $c_3>0$ is a constant depending only on $ c_1, c_2$.
\end{theorem}
\begin{proof}
Set
\[
\widehat{\mathcal{R}}\,\,:=\,\,\dkh{\vz\in \C^d: \norm{\vz}\le 5 \norm{\vx} }.
\]
  Lemma \ref{le:Lipcondition} implies that, with probability at least $1-4\exp(-cm)-c_a m^{-d}$,  the Wirtinger gradient $\nabla f(\vz)$ given in \eqref{eq:wirtgrad} is Lipschitz continuous over $\widehat{\mathcal{R}}$, namely,
\begin{equation}  \label{le:firssmoo}
\norm{\nabla f(\vz) -\nabla f(\vz')} \le  L_{\widehat{\mathcal{R}}}  \norm{\vz-\vz'} \quad \mbox{for all} \quad \vz,\vz' \in \widehat{\mathcal{R}},
\end{equation}
where $L_{\widehat{\mathcal{R}}}:=C_1 d\log m \norm{\vx}^2$ for a constant $C_1>0$ depending only on $c_1$ and $c_2$.
 Here, we use $\norm{\vb} \le \sqrt[4]{m \sum_{j=1}^m |b_j|^4} \le \sqrt[4] c_1 \sqrt m \norm{\vx}$ which follows from  $ \sum_{j=1}^m |b_j|^4 \le c_1 m \norm{\vx}^4 $ and the Cauchy-Schwarz inequality.
 We next claim that, with probability at least $1-4\exp(-cm)-c_a m^{-d}$, it holds that
\begin{equation} \label{eq:mainLip}
f(\vz')\le f(\vz) +  2\nj{\nabla f(\vz), \vz'-\vz}_{\Re} + \frac{\sqrt 2 L_{\widehat{\mathcal{R}}} }{2} \norm{\vz'-\vz}^2 \quad \mbox{for all} \quad \vz',\vz \in  \widehat{\mathcal{R}}.
\end{equation}
Here, $\nabla f(\vz)$ is given in \eqref{eq:abugrad}.

 Theorem \ref{th:formmain} implies that the loss function  $f(\vz)$ is strongly convex with probability at least  $1-c_a m^{-1}-c_b \exp(-c_c  m /\log m)-18 \exp(-c_d d)$. Hence, we have
\begin{equation} \label{eq:mainstrcon}
f(\vx) \ge f(\vz) + 2\nj{\nabla f(\vz), \vx-\vz}_{\Re} + \frac{\beta}{2}  \norm{\vz-\vx}^2 \quad \mbox{for all} \quad \vz \in \C^d,
\end{equation}
where $\beta:=1.96c_0^2-4.4$ (see  Remark \ref{Re:relastrongcon} for detail).

Based on \eqref{eq:mainLip} and \eqref{eq:mainstrcon}, we can prove the conclusion recursively. Indeed, since the initial point $\vz_0 \in \mathbb{B}_{\C}^d(R_0)$,  we have $\norm{\vz_0} \le R_0$.
  According to Lemma \ref{le:bound4x}, we  obtain that, with probability at least $1-4\exp(-c m)$, it holds  that $\norm{\vx} \le R_0 \le 3\norm{\vx}$, which implies $\vz_0,\vx \in \widehat{\mathcal{R}}$.
Next, if we assume $\vz_k \in \widehat{\mathcal{R}}$ then
\begin{eqnarray}
\norm{\vz_{k+1}-\vx}^2 & =& \norm{\vz_k -\vx -\mu \nabla f(\vz_k)}^2\nonumber \\
&= &  \norm{\vz_{k}-\vx}^2+2\mu\cdot \nj{\nabla f(\vz_k), \vx-\vz_k}_{\Re} +\mu^2 \norm{ \nabla f(\vz_k)}^2 \nonumber \\
&\le & \xkh{1-\frac{\mu\beta}{2}}  \norm{\vz_{k}-\vx}^2-\mu\xkh{f(\vz_k)-f(\vx) }+\mu^2 \norm{ \nabla f(\vz_k)}^2\nonumber\\
&\le & \xkh{1-\frac{\mu\beta}{2}}  \norm{\vz_{k}-\vx}^2-\mu\xkh{f(\vz_k)-f(\vx) }+\frac{2\mu^2 L_{\widehat{\mathcal{R}}} }{4-\sqrt 2}  \xkh{f(\vz_k)-f(\vx) }\nonumber \\
&=&  \xkh{1-\frac{\mu\beta}{2}}  \norm{\vz_{k}-\vx}^2-\mu \xkh{1- \frac{2\mu L_{\widehat{\mathcal{R}}} }{4-\sqrt 2}   } \xkh{f(\vz_k)-f(\vx) }\nonumber \\
&\le & \xkh{1-\frac{\mu\beta}{2}}  \norm{\vz_{k}-\vx}^2, \label{eq:recurcont}
\end{eqnarray}
provided the step size $\mu \le (4-\sqrt 2)/(2L_{\widehat{\mathcal{R}}})$, where the first inequality follows from \eqref{eq:mainstrcon} and the second inequality follows from the fact of
\begin{eqnarray*}
f(\vz_k)-f(\vx) & \ge &  f(\vz_k)-f(\vz_k- \frac 1{ L_{\widehat{\mathcal{R}}}} \nabla f(\vz_k))\\
&  \ge & f(\vz_k) -  \left( f(\vz_k) +  2\nj{\nabla f(\vz_k), - \frac 1{ L_{\widehat{\mathcal{R}}}} \nabla f(\vz_k)}_{\Re} + \frac{\sqrt 2 L_{\widehat{\mathcal{R}}} }{2} \norm{ \frac 1{ L_{\widehat{\mathcal{R}}}} \nabla f(\vz_k)}^2\right) \\
&  = & \frac{4-\sqrt 2}{2L_{\widehat{\mathcal{R}}}  }  \norm{ \nabla f(\vz_k)}^2.
\end{eqnarray*}
 Here, we use the fact that $f(\vz)\geq f(\vx)$ for any $\vz\in \C^d$ in the first inequality and the claim \eqref{eq:mainLip} in the second inequality. We can use\eqref{eq:recurcont} to obtain that
\[
\norm{ \vz_{k+1} } \le \norm{\vz_{k+1}-\vx} +\norm{\vx} \le  \norm{\vz_0-\vx} +\norm{\vx} \le  \norm{\vz_{0}} +2\norm{\vx} \le 5\norm{\vx},
\]
which implies  $\vz_{k+1} \in \widehat{\mathcal{R}}$. Applying the  \eqref{eq:recurcont} recursively and observing that $\norm{\vz_0-\vx} \le 4\norm{\vx}$ holds with probability at least $1-4\exp(-c m)$, we arrive at the conclusion.

It remains to prove the claim \eqref{eq:mainLip}. Indeed, from the fundamental theorem of calculus, as shown in \eqref{eq:strrele}, we have
\begin{eqnarray*}
f(\vz') &=& f(\vz)+\int_0^1  \xkh{\nabla f(\vz+ t (\vz'- \vz) )}^* \left [ \begin{array}{l} \vz'- \vz  \vspace{0.5em}\\ \overline{\vz'- \vz} \end{array}\right] dt \nonumber \\
&\le &  f(\vz)+ \xkh{\nabla f(\vz)}^* \left [ \begin{array}{l} \vz'- \vz  \vspace{0.5em}\\ \overline{\vz'- \vz } \end{array}\right] + \sqrt 2 \int_0^1\norm{ \nabla f(\vz+ t (\vz'- \vz) ) -\nabla f(\vz)}\norm{\vz'-\vz}  dt\\
&\le &  f(\vz) +  2\nj{\nabla f(\vz), \vz'-\vz}_{\Re} +\sqrt 2 L_{\widehat{\mathcal{R}}}   \int_0^1 t \norm{\vz'-\vz}^2~dt \\
&\le &  f(\vz) +  2\nj{\nabla f(\vz), \vz'-\vz}_{\Re} +  \frac{\sqrt 2 L_{\widehat{\mathcal{R}}} }{2} \norm{\vz'-\vz}^2,
\end{eqnarray*}
where the second inequality follows from \eqref{le:firssmoo}. This completes the proof of the claim \eqref{eq:mainLip}.
\end{proof}

\section{Numerical Simulations}
In this section, we demonstrate experimentally that the objective function $f$ given in (\ref{eq:fz}) is well structured even when the number of measurements $m=O(d)$. To this end, we test  the efficiency and robustness of Algorithm \ref{Al:1} via a series of numerical experiments.  In our numerical experiments, the target vector $\vx\in \C^d$ is chosen randomly from the standard complex Gaussian  distribution, that is $\vx \sim  \mathcal{N}(0,I_d)+i  \mathcal{N}(0,I_d)$. The measurement vectors $ \va_j, \,j=1,\ldots,m$ are generated randomly from standard complex Gaussian distribution, and the bias vector $\vb \sim 5\norm{\vx} \cdot \mathcal{N}(0,I_d)$.
All experiments are carried out on a laptop computer with a 2.4GHz Intel Core i7 Processor, 8 GB 2133 MHz LPDDR3 memory and Matlab R2016a.

\begin{example}
In this example, we test the empirical success rate of the Algorithm \ref{Al:1} versus the number of measurements.  We set $d=1000$ and vary $m$ within the range $[3d,8d]$. The step size $\mu=0.01$.
 For each $m$, we run $100$ times trials to calculate the success rate.  Here, we say a trial to have successfully reconstructed the target signal if the algorithm returns a vector $\vz_{T}$ which has a small relative error, that is when $\norm{\vz_{T}-\vx}/\norm{\vx} \le 10^{-5}$.  The results are plotted in Figure \ref{figure:succ}.  It can be seen that the Algorithm \ref{Al:1} achieves 100\% success rate when the number of measurements $m\ge 6.5 d$, which means $m\ge 6.5 d$ samples may be sufficient to ensure the strong convexity property holds.
\end{example}

\begin{figure}{}
\centering
     \includegraphics[width=0.45\textwidth]{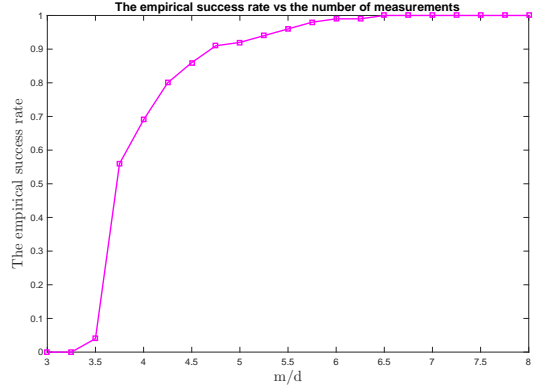}
\caption{ The empirical success rate for different $m/d$ based on $100$ random trails. }
\label{figure:succ}
\end{figure}

\begin{example}
{ In this example, we test the convergence rate of the Algorithm \ref{Al:1}. We choose $d=1000$ and $m=7d$. The step size $\mu=0.01$.  To show the robustness, we  consider the noisy data model $y_j=\abs{\nj{\va_j,\vx}+b_j}^2+\eta_j$ where the noise $\eta_j\sim  {\mathcal N}(0, 0.01^2)$. The results are presented in Figure \ref{figure:relative_error}, which verifies the linear convergence of the Algorithm \ref{Al:1}.

\begin{figure}{}
\centering
\subfigure[]{
     \includegraphics[width=0.45\textwidth]{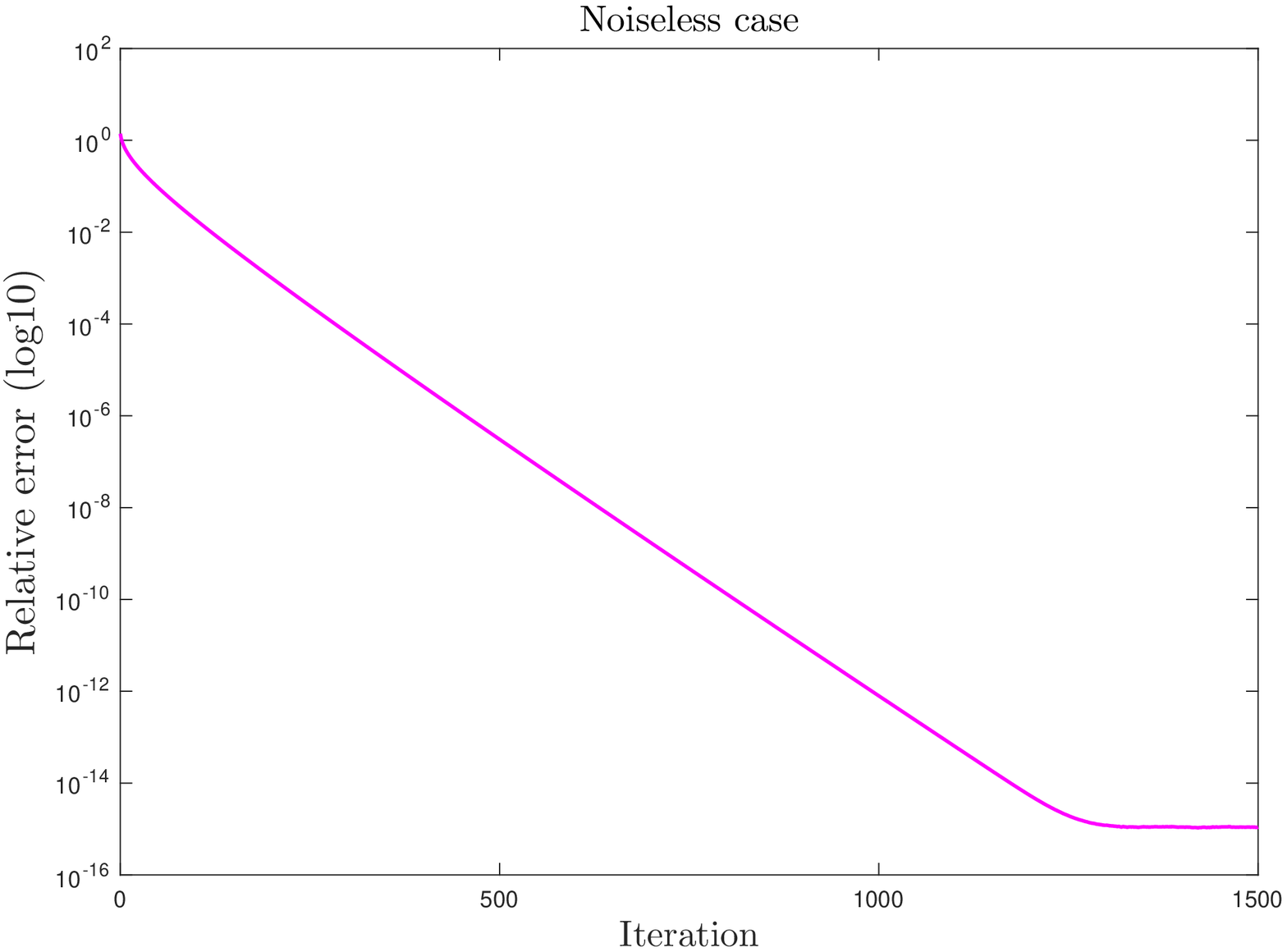}}
\subfigure[]{
     \includegraphics[width=0.45\textwidth]{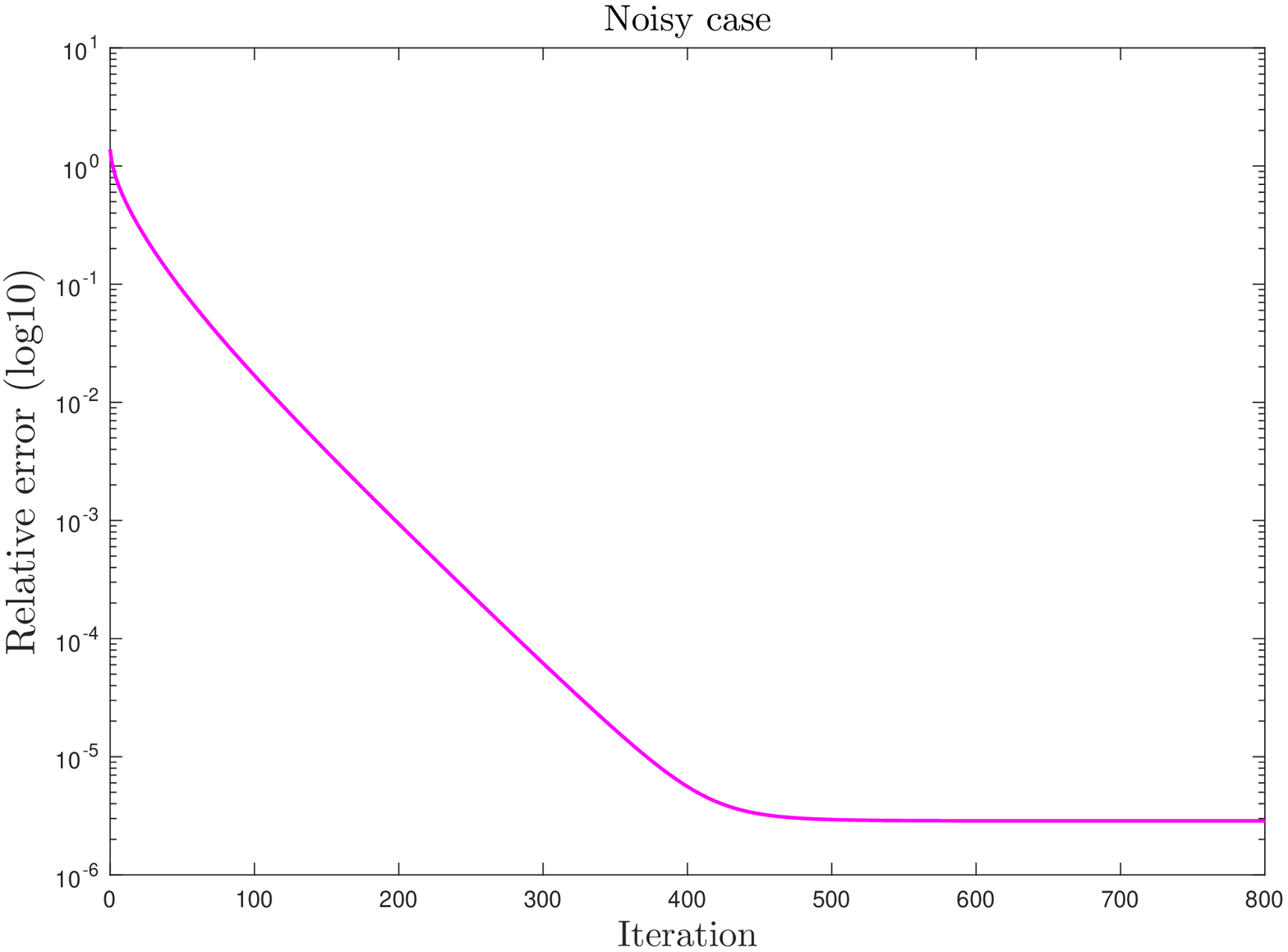}}
\caption{ Relative error versus the number of iterations: (a) The noiseless measurements; (b) The noisy measurements.}
\label{figure:relative_error}
\end{figure}
}
\end{example}

\section{Discussion}
In this paper, we provide the characterization of a natural least squares formulation \eqref{eq:fz} for the affine phase retrieval problem. We show the loss function $f$ given in   \eqref{eq:fz} is strongly convex on the entire space $\C^d$. This benign geometric structure allows the simple gradient descent algorithm to reconstruct the target signals with linear convergence rate.

There are some interesting problems for future research. First,  Theorem  \ref{th:formmain} requires $m \gtrsim d\log d$ samples to guarantee the strongly convexity.
 Based on numerical experiments, we conjecture $m \gtrsim d$ samples are sufficient to ensure the property. It is interesting to see whether the gap can be closed. Second, our current analysis requires the measurements to be Gaussian random vectors. It is of practical interest to extend this result to other  measurements, such as sub-Gaussian measurements, Fourier measurements and short-time Fourier measurements etc.

\appendix

\section{Preliminaries and supporting lemmas}
\begin{lemma}[Chebyshev's inequality] \label{chebv}
For a random variable $X$ with finite variance $\sigma^2=\mbox{\rm Var}( X)$, we have
\[
\PP\xkh{\abs{X- \E X} \ge t } \le \frac{\sigma^2}{t^2} \quad \mbox{for all} \quad t>0.
\]
\end{lemma}

\begin{lemma}\cite[Bernstein's inequality]{Vershynin2018} \label{le:Ber}
Let $X_1,\ldots,X_m$ be independent, mean zero, sub-exponential random variables, and $\vb=(b_1,\ldots,b_m) \in \R^m$.
 Then  for any $t\ge 0$ it holds that
\[
\PP\dkh{\Big | \sum_{j=1}^m b_j X_j\Big| \ge t} \le 2\exp\xkh{-c \min \Big(\frac{t^2}{K^2 \norm{\vb}^2} , \frac{t}{K \|\vb\|_{\infty}}\Big) },
\]
where $c>0$ is an absolute constant and  $K$ is the maximal  sub-exponential norm, i.e.,  $K:= \mathrm{max}_j \| X_j\|_{\psi_1}$.
\end{lemma}

The following result is a complex version of Lemma 4.4.3 in \cite{Vershynin2018} and the proof is the same as that of Lemma 4.4.3 in \cite{Vershynin2018}.

\begin{lemma}  \label{le:supnorm}
Let $M\in \C^{d\times d}$  be a Hermitian matrix and $\varepsilon \in [0,1/2)$. Then we have
\[
\norm{M} \le \frac 1{1-2\varepsilon}\cdot \sup_{\vx\in \mathcal{N}} |\nj{M\vx,\vx}|,
\]
where  $\mathcal{N}$ is a $\varepsilon$-net  of the sphere $\mathbb{S}_{\C}^{d-1}$.
\end{lemma}

\begin{lemma}\cite[Lemma 21]{turstregion} \label{le:sunju}
Let $\va_j \in \C^d, j=1,\ldots,m,$ be i.i.d complex Gaussian random vectors. Suppose
that  $\vx \in \C^d$ is a fixed vector. For any $\epsilon \in (0,1)$ the following holds with
probability at least $1-c_a \epsilon^{-2} m^{-1}-c_b \exp(-c_c \epsilon^2 m /\log m)$:
\[
\left\| \frac{1}{m} \sum_{j=1}^m \abs{\va_j^* \vx}^2  \va_j \va_j^*- \big( \vx\vx^* +\norm{\vx}^2 ~ I \big)    \right\|          \le  \epsilon \norms{\vx}^2
\]
provided $m\ge C(\epsilon) d \log d$. Here  $C(\epsilon)$ is a constant depending on
$\epsilon$, $c_a$, $c_b$ and $c_c$ are positive absolute constants.
\end{lemma}

\begin{lemma}\cite[Lemma 22]{turstregion} \label{le:sunju22}
Let $\va_j \in \C^d, j=1,\ldots,m,$ be i.i.d complex Gaussian random vectors. For any
$\epsilon \in (0,1)$  the followings hold  with probability at least $1-c'  m^{-d}-c''
\exp(-c(\epsilon) m )$:
\[
\frac{1}{m} \sum_{j=1}^m \abs{\va_j^* \vz}^2 \abs{\va_j^* \bm{v}}^2  \ge (1-\epsilon) \xkh{\norms{\bm{v}}^2\norms{\vz}^2+\abs{\bm{v}^*\vz}^2} \quad \mbox{for all} \quad \vz,\bm{v} \in \C^d,
\]
\[
\frac{1}{m} \sum_{j=1}^m \zkh{(\va_j^* \vz)(\vv^* \va_j)}_{\Re}^2 \ge (1-\epsilon) \xkh{\frac12 \norm{\vz}^2 \norm{\vv}^2 +\frac32 \zkh{\Re (\vz^*\vv)}^2-\frac12 \zkh{\Im (\vz^*\vv)}^2}
\]
\[
\mbox{for all} \quad \vz,\; \vv \in \C^d,
\]
provided $m\ge C(\epsilon) d \log d$.
 Here  $C(\epsilon)$ and $c(\epsilon)$ are constants depending on $\epsilon$ and
$c'$, $c''$ are positive absolute constants.
\end{lemma}

The following lemma is an alternative  version of Lemma 3.3 in \cite{huang2021a}.
\begin{lemma} \label{le:etalowbound}
Let $\eta \in \R^m$ be a fixed vector.
Suppose that $\va_j \in \C^d, j=1,\ldots,m $, are i.i.d. complex Gaussian random vectors.
Then there exists a universal constant $C>0$ such that the following holds with probability at least $1-2\exp(-c_0  d)$:
 \[
   \norm{ \sum_{j=1}^m \eta_j (\va_j\va_j^*-I)} \le C\xkh{ \sqrt{d}\norm{\eta}+d \norms{\eta}_{\infty}}.
 \]
Here, $c_0>0$ is a universal constant.
\end{lemma}

\begin{lemma}\label{le:dizwbj}
Assume that $\va_j \in \C^d, j=1,\ldots,m,$ are i.i.d complex Gaussian random vectors. If $m\ge Cd$ for a universal constant $C>0$ then, with probability at least $1-3\exp(-c m)$, the following  holds:
\[
\frac 1m \sum_{j=1}^m \abs{|\va_j^* \vz+b_j|^2-|\va_j^* \vw+b_j|^2 } \le \frac 32 \xkh{  \norm{\vz}+\norm{\vw}+ \frac{2 \norm{\vb}}{\sqrt m } } \norm{\vz-\vw}
\]
for all $\vz,\vw \in \C^d$ and  $\vb=(b_1,\ldots,b_m)^\T \in \C^m$.  Here, $c$ is a positive absolute constants.
\end{lemma}
\begin{proof}
A simple calculation shows that
\begin{eqnarray}
&& \frac 1m \sum_{j=1}^m \abs{|\va_j^* \vz+b_j|^2-|\va_j^* \vw+b_j|^2 } \nonumber\\
 &\quad &\le \frac 1m \sum_{j=1}^m \abs{|\va_j^* \vz|^2-|\va_j^* \vw|^2 }+\frac 2m \sum_{j=1}^m |b_j| \abs{\va_j^*(\vz-\vw)} \nonumber\\
& \quad &\le \frac 1m \sum_{j=1}^m \abs{|\va_j^* \vz|^2-|\va_j^* \vw|^2 } + \frac{2 \norm{\vb}}{\sqrt m } \sqrt{\frac 1m \sum_{j=1}^m |\va_j^*(\vz-\vw)|^2 }.  \label{eq:difzzww}
\end{eqnarray}

We first consider the  term  $\frac 1m \sum_{j=1}^m \abs{|\va_j^* \vz|^2-|\va_j^* \vw|^2 }$ in (\ref{eq:difzzww}). According to  Lemma 3.1 in \cite{phaselift}, the following  holds with probability at least $1-2\exp(-c m)$:
\begin{eqnarray}
 \frac 1m \sum_{j=1}^m \abs{|\va_j^* \vz|^2-|\va_j^* \vw|^2 } &\le &  \frac{3}{2\sqrt 2} \Big\| \vz\vz^*-\vw\vw^*\Big\|_* \nonumber\\
 &\le & \frac 32 \norm{\vz\vz^*-\vw\vw^*} \nonumber\\
 & \le & \frac 32 (\norm{\vz}+\norm{\vw})\norm{\vz-\vw}. \label{eq:nufzw}
\end{eqnarray}
for all $ \vz,\vw \in \C^d$. Here, $\|\cdot\|_*$ denotes the nuclear norm.

We next turn to the second term in (\ref{eq:difzzww}).  We have
\begin{equation} \label{eq:aatran}
\frac 1m \sum_{j=1}^m |\va_j^*(\vz-\vw)|^2 \le  \Big\| \frac 1m \sum_{j=1}^m \va_j\va_j^*\Big\|_2 \Big\| \vz-\vw\Big\|_2^2 \le \frac 94 \Big\| \vz-\vw\Big\|_2^2,
\end{equation}
where we use the fact that $ \Big\| \frac 1m \sum_{j=1}^m \va_j\va_j^*\Big\|_2 \le \frac 94$ with probability at least $1-\exp(-cm)$ in the last inequality.

Putting \eqref{eq:aatran} and \eqref{eq:nufzw} into \eqref{eq:difzzww}, we obtain  that when $m\ge Cd$ then with probability at least $1-3\exp(-c m)$ it holds
\[
 \frac 1m \sum_{j=1}^m \abs{|\va_j^* \vz+b_j|^2-|\va_j^* \vw+b_j|^2 }  \le  \frac 32 \xkh{  \norm{\vz}+\norm{\vw}+ \frac{2 \norm{\vb}}{\sqrt m } } \norm{\vz-\vw}
\]
for all $\vz,\vw \in \C^d$.
\end{proof}

\begin{lemma} \label{le:tail}
 Suppose that $\beta  \ge 1$ is fixed constant.   Assume $\va_j \in \C^d, j=1,\ldots,m $, are i.i.d. complex Gaussian random vectors and
$\vb \in \C^m$ obeys $ \sum_{j=1}^m  |b_j|^4 \lesssim m$ and $\norms{\vb}_{\infty} \le \sqrt{\log m}$.
For any $ \epsilon\in (0, 1)$, if $m\ge C\epsilon^{-2}\log (1/\epsilon) d \log m$ then with probability at least $1-2\exp(-c' \epsilon^2 m/ \log m)-2\exp(-c'' d) $  it holds that
\[
 \frac 1m \sum_{j=1}^m |b_j|^2 |\va_j^* \vv|^2\1_{|\va_j^* \vv| \ge \beta}  \le \xkh{2\beta e^{-0.49 \beta^2} +\epsilon}  \cdot \frac{\norm{\vb}^2}{m}  +  \epsilon
\]
 for all $\vv \in \mathbb{S}_{\C}^{d-1}$. Here, $C, c', c'' >0$ are universal constants and $~\mathbb{S}_{\C}^{d-1}:=\{\vv\in \C^d: \|\vv\|=1\}$.
\end{lemma}

\begin{proof}
Due to the non-Lipschitz of the indicator function $\1_{|\va_j^* \vv| \ge \beta}$, we introduce an auxiliary Lipschitz function to approximate it.
Set
\[
\chi (t):=\left\{ \begin{array}{ll}
                         t,& \mbox{if} \quad t \ge \beta^2 ; \\
                         \frac 1{\delta} t -(\frac 1{\delta}-1)\beta^2 ,& \mbox{if} \quad  (1-\delta)\beta^2 \le t \le \beta^2;\\
                         0,& \mbox{otherwise}.
                         \end{array}
                         \right.
\]
Here $\delta \in (0,1)$ is a constant which will be chosen later.
Then it is easy to check that
\begin{equation}\label{eq:chithesec}
 \frac 1m \sum_{j=1}^m |b_j|^2 |\va_j^* \vv|^2\1_{|\va_j^* \vv| \ge \beta}  \le  \frac{1}{m} \sum_{j=1}^m  |b_j|^2  \chi (|\va_j^* \vv|^2).
  %\le  \frac 1m \sum_{j=1}^m |b_j|^2 |\va_j^* \vv|^2\1_{|\va_j^* \vv| \ge (1-\delta)\beta}.
\end{equation}
For any fixed $\vv_0\in \mathbb{S}_{\C}^{d-1}$,  the terms $\chi(|\va_j^* \vv_0|^2) $ are independent sub-exponential random variables with the maximal sub-exponential  norm
 being a constant. According to Bernstein's inequality (Lemma \ref{le:Ber}), for any fixed $t \ge 0$ , it holds that
\begin{equation}\label{eq:bern1}
\PP\dkh{ \abs{  \frac{1}{m} \sum_{j=1}^m  |b_j|^2  \chi (|\va_j^* \vv_0|^2)  - \frac{\norm{\vb}^2}{m} \cdot \E\zkh{ \chi(|\va_1^* \vv_0|^2)  }} \ge t} \le 2\exp\xkh{-c \min\Big(\frac{m^2 t^2}{ \sum_{j=1}^m  |b_j|^4},\frac{m t}{\norms{\vb}_{\infty}^2}\Big)},
\end{equation}
where $c>0$ is a universal constant.  Recall that
$
 \sum_{j=1}^m  |b_j|^4 \lesssim m \quad  \mbox{and} \quad \norms{\vb}_{\infty}^2 \lesssim  \log m
 $.
 For any $0< \epsilon <1$, taking $t:=\epsilon/2$  in (\ref{eq:bern1}), we obtain that
\begin{equation} \label{eq:bern2}
\frac{1}{m} \sum_{j=1}^m  |b_j|^2  \chi (|\va_j^* \vv_0|^2) \le \frac{\norm{\vb}^2}{m} \cdot \E\zkh{ \chi(|\va_1^* \vv_0|^2) } +  \frac{\epsilon}{2}
\end{equation}
holds with probability at least $1-2\exp(-c_1 \epsilon^2 m/\log m)$, where $c_1>0$ is a universal constant.

We next show that (\ref{eq:bern2}) holds for any $\vv\in {\mathbb S}_{\C}^{d-1}$.
Suppose that $\mathcal{N}$ is a $\varepsilon_0$-net  over ${\mathbb S}_{\C}^{d-1}$ with  $\#\mathcal{N} \le (1+\frac{2}{\varepsilon_0})^{2d}$.
Then for any $\vv\in {\mathbb S}_{\C}^{d-1}$, there exists a $\vv_0 \in \mathcal{N}$ such that $\norm{\vv-\vv_0}\le \varepsilon_0$.
Note that $\chi(t)$ is a Lipschitz function with Lipschitz  constant $1/\delta$.  We obtain that if $m\ge C_1 d\log m$ for a universal constant $C_1>0$ then
\begin{equation}\label{eq:neteq}
\begin{aligned}
& \abs{ \frac{1}{m} \sum_{j=1}^m  |b_j|^2  \chi (|\va_j^* \vv|^2)  -\frac{1}{m} \sum_{j=1}^m  |b_j|^2  \chi (|\va_j^* \vv_0|^2) } \\
&\le   \frac{1}{m\delta} \sum_{j=1}^m  |b_j|^2  \abs{|\va_j^* \vv|^2-|\va_j^* \vv_0|^2 } \\
&\le   \frac{1}{m\delta} \sum_{j=1}^m  |b_j|^2  |\va_j^* \vv| \abs{ \va_j^* (\vv-\vv_0)}+\frac{1}{m\delta} \sum_{j=1}^m  |b_j|^2  |\va_j^* \vv_0| \abs{ \va_j^* (\vv-\vv_0)} \\
&\le  \frac 1 {\delta}\xkh{\sqrt{\frac{1}{m} \sum_{j=1}^m  |b_j|^2  |\va_j^* \vv|^2}+\sqrt{\frac{1}{m} \sum_{j=1}^m  |b_j|^2  |\va_j^* \vv_0|^2}} \cdot \sqrt{ \frac{1}{m} \sum_{j=1}^m  |b_j|^2 \abs{ \va_j^* (\vv-\vv_0)}^2} \\
&\overset{(a)}\lesssim  \frac 1{\delta}\xkh{\frac{\norm{\vb}^2}{m} + \sqrt{\frac 1 m \sum_{j=1}^m  |b_j|^4 }+\frac{d \norms{\vb}_{\infty}^2} m } \varepsilon_0\\
&\lesssim  \frac 1{\delta} \xkh{\frac{\norm{\vb}^2}{m} + 1} \varepsilon_0.
\end{aligned}
\end{equation}
Here, the inequality  $(a)$  follows from Lemma \ref{le:etalowbound} which says that, with probability at least $1-2\exp(-c_2 d)$, it holds that
 \[
 \sum_{j=1}^m  |b_j|^2  |\va_j^* \vw|^2 \le \vw^* \xkh{ \sum_{j=1}^m  |b_j|^2  \va_j\va_j^*}\vw \lesssim \xkh{\norm{\vb}^2 + \sqrt d \sqrt{  \sum_{j=1}^m  |b_j|^4} +d\norms{\vb}_{\infty}^2}\norm{\vw}^2
 \]
 for all $\vw \in \C^d$, where $c_2>0$ is a universal constant. Here, we use the fact that $ \sum_{j=1}^m  |b_j|^4 \lesssim m$ and $ \norms{\vb}_{\infty}^2 \lesssim  \log m$.
  Choosing $\varepsilon_0:=c_3 \epsilon \delta$ in (\ref{eq:neteq}) for some universal constant $c_3>0$ and taking the union bound over $\mathcal N$, we obtain that
\begin{equation} \label{eq:unibjx}
\frac{1}{m} \sum_{j=1}^m  |b_j|^2  \chi (|\va_j^* \vv|^2) \le \frac{\norm{\vb}^2}{m} \cdot  \xkh{ \E\zkh{ \chi(|\va_1^* \vv|^2) } +\epsilon}+ \epsilon
\end{equation}
holds for all $\vv \in \mathbb{S}_{\C}^{d-1}$ with probability at least
\[
1-2(1+\frac{2}{\varepsilon_0})^{2d} \cdot  \exp(-c_1 \epsilon^2 m/\log m)-2\exp(-c_2 d) \ge 1-2\exp(-c' \epsilon^2 m/\log m)-2\exp(-c_2 d)
\]
provided $m\ge C' \epsilon^{-2}\log (\epsilon^{-1}\delta^{-1}) ~d \log m$, where $C', c'>0$ are universal constants.

We next turn to  $\E\zkh{ \chi(|\va_1^* \vv|^2) }$ in (\ref{eq:unibjx}).
A simple calculation shows that
\begin{equation} \label{eq:expavgam}
\E\xkh{|\va_1^* \vv|^2\1_{|\va_j^* \vv| \ge \gamma}} \le \sqrt{\frac{2}{\pi}} \xkh{\gamma+\frac1{\gamma}}e^{-\gamma^2/2}
\end{equation}
for any $\gamma>0$.  Note that $\E\zkh{ \chi(|\va_1^* \vv|^2) }  \le \E\zkh{ |\va_1^* \vv|^2 \1_{|\va_j^* \vv| \ge (1-\delta)\beta}}$.  Set $\delta:=0.01$ and $\gamma:=(1-\delta)\beta$.  It then follows from \eqref{eq:chithesec}, \eqref{eq:unibjx} and  \eqref{eq:expavgam} that with probability at least $1-2\exp(-c' \epsilon^2 m/\log m)-2\exp(-c_2 d) $, it holds
\[
 \frac 1m \sum_{j=1}^m |b_j|^2 |\va_j^* \vv|^2\1_{|\va_j^* \vv| \ge \beta}  \le \xkh{2\beta e^{-0.49 \beta^2} +\epsilon}  \cdot \frac{\norm{\vb}^2}{m}  +  \epsilon
\]
 for all $\vv \in \mathbb{S}_{\C}^{d-1}$,  provided $m\ge C\epsilon^{-2}\log (1/\epsilon)~ d\log m$. Here, $C>0$ is a universal constant.  We complete the proof.

\end{proof}

\begin{lemma} \label{le:azav2}
Suppose that $\va_j \in \C^d,j=1,\ldots,m $, are i.i.d. complex Gaussian random vectors.
Assume that $\vb \in \C^m$ is a vector obeying  $ \sum_{j=1}^m  |b_j|^4 \lesssim m$ and $\norms{\vb}_{\infty} \lesssim \sqrt{\log m}$.
For any $\epsilon \in (0,1)$, if $m\ge  c(\epsilon) d \log m $ then the following holds
with probability at least $1-6 \exp(-c'(\epsilon) m/\log m ) - 6\exp(-c'' d)$:
\[
\abs{\frac 1m \sum_{j=1}^m  \zkh{\bar{b}_j (\va_j^* \vz)}_{\Re} |\va_j^* \vv|^2} \le  \epsilon\xkh{ \frac{\norm{\vb}}{\sqrt m}+1} \xkh{1 + \xkh{\frac 1m \sum_{j=1}^m |\va_j^* \vz|^2 |\va_j^* \vv|^2}^{\frac 12}  }
\]
for all $ \vz, \vv \in \mathbb{S}_{\C}^{d-1}$, where  $~\mathbb{S}_{\C}^{d-1}:=\{\vx\in \C^d: \|\vx\|=1\}$, $c''$ is a positive universal constant and $c(\epsilon)$, $c'(\epsilon)$ are positive constants depending only on $\epsilon$.
\end{lemma}
\begin{proof}
Suppose that $\phi\in C_c^{\infty}(\mathbb R)$ is a Lipschitz continuous function
satisfying $0\le \phi(x)\le 1$ for all $x\in \mathbb R$. We furthermore require
$\phi(x)=1$ for $|x|\le 1$ and $\phi(x)=0$ for $|x|\ge 2$. For any   $\beta>0$, we
have
\begin{equation}\label{eq:T1T2}
\begin{aligned}
&\abs{\frac 1m \sum_{j=1}^m  \zkh{\bar{b}_j (\va_j^* \vz)}_{\Re} |\va_j^* \vv|^2}  \\
&\le  \abs{\frac 1m \sum_{j=1}^m  \zkh{\bar{b}_j (\va_j^* \vz)}_{\Re} |\va_j^* \vv|^2 \phi\xkh{\frac {|\va_j^* \vv|} {\beta}} }+\abs{ \frac 1m \sum_{j=1}^m  \zkh{\bar{b}_j (\va_j^* \vz)}_{\Re} |\va_j^* \vv|^2 \xkh{1-\phi\xkh{\frac {|\va_j^* \vv|} {\beta}} } }\\
&\le  \abs{\frac 1m \sum_{j=1}^m b_{j,\Re} (\va_j^* \vz)_{\Re} |\va_j^* \vv|^2 \phi\xkh{\frac {|\va_j^* \vv|} {\beta}}} +\abs{ \frac 1m \sum_{j=1}^m b_{j,\Im} (\va_j^* \vz)_{\Im} |\va_j^* \vv|^2 \phi\xkh{\frac {|\va_j^* \vv|} {\beta}}} \\
&\quad + \frac 1m \sum_{j=1}^m |b_j| |\va_j^* \vz| |\va_j^* \vv|^2\1_{\dkh{|\va_j^* \vv | \ge \beta}}   := T_1+T_2 + r.
\end{aligned}
\end{equation}
%where
%\begin{equation*}
%\begin{aligned}
% T&:= \frac 1m \sum_{j=1}^m |\va_j^* \vz|^2 \Re(\vx_0^* \va_j \va_j^* \vz) \phi\xkh{\frac {|\va_j^* \vz|} {\beta}},\\
%r&:=\frac 1m \sum_{j=1}^m |\va_j^* \vz|^2 \Re(\vx_0^* \va_j \va_j^* \vz)
% \xkh{1- \phi\xkh{\frac {|\va_j^* \vz|}{\beta}}} \le \frac 1m \sum_{j=1}^m \abs{\va_j^* \vz}^3  \abs{\va_j^* \vx_0} \1_{\dkh{|\va_j^* \vz| \ge \beta}}.
% \end{aligned}
%\end{equation*}
We claim that  for any $\epsilon\in (0, 1)$ there exists a  sufficiently large $\beta>1$
such that   if $m\ge c(\epsilon) d \log m $ then, with probability
at least $1-6 \exp(-c'(\epsilon) m/\log m ) - 6\exp(c'' d)$, the followings hold
\begin{eqnarray}\label{eq:claim0}
T_1\le  \frac{\epsilon}{2} \cdot\xkh{ \frac{\norm{\vb}}{\sqrt m}+1},\quad T_2\le  \frac{\epsilon}{2} \cdot \xkh{ \frac{\norm{\vb}}{\sqrt m}+1},\\
r \le  \epsilon \xkh{ \frac{\norm{\vb}}{\sqrt m}+1} \xkh{\frac 1m \sum_{j=1}^m |\va_j^* \vz|^2 |\va_j^* \vv|^2}^{\frac 12}\label{eq:claim1}
\end{eqnarray}
for all $ \vz, \vv  \in \mathbb{S}_{\C}^{d-1}$.
 Here
 $c(\epsilon), c'(\epsilon)  $ are constants depending on $\epsilon$ and $c''$ is a positive universal constant.  Combining (\ref{eq:T1T2}),   (\ref{eq:claim0}) and (\ref{eq:claim1}), we arrive at the conclusion, i.e,
\[
\abs{\frac 1m \sum_{j=1}^m  \zkh{\bar{b}_j (\va_j^* \vz)}_{\Re} |\va_j^* \vv|^2} \le  \epsilon\xkh{ \frac{\norm{\vb}}{\sqrt m}+1} \xkh{1 + \xkh{\frac 1m \sum_{j=1}^m |\va_j^* \vz|^2 |\va_j^* \vv|^2}^{\frac 12}  }
\]
holds for all $ \vz, \vv \in \mathbb{S}_{\C}^{d-1}$.

It remains to prove the claims (\ref{eq:claim0}) and (\ref{eq:claim1}). We first present an upper bound for $T_1$. For any fixed $\vz_0,\vv_0 \in \mathbb{S}_{\C}^{d-1}$, due to the cut-off $\phi\xkh{\frac {|\va_j^* \vv|}{\beta}}$, the terms
$ (\va_j^* \vz_0)_{\Re} |\va_j^* \vv_0|^2 \phi\xkh{\frac {|\va_j^* \vv_0|} {\beta}}$ are centered,
independent sub-gaussian random variables with the sub-gaussian norm
$O(\beta)$. According to Hoeffding's inequality, we obtain that the following holds
with probability at least $1-2\exp(-c  \epsilon^2 \beta^{-2} m)$
\begin{equation}\label{eq:gudvz00}
\abs{\frac 1m \sum_{j=1}^m  b_{j,\Re} (\va_j^* \vz_0)_{\Re} |\va_j^* \vv_0|^2 \phi\xkh{\frac {|\va_j^* \vv_0|} {\beta}}}  \le \frac {\epsilon \norm{\vb}}{4 \sqrt{m}},
\end{equation}
 where $c >0$ is a universal constant.
 We next show that (\ref{eq:gudvz00}) holds
for all unit vectors $\vz, \vv \in {\mathbb S}_{\C}^{d-1}$.  We adopt a basic version of a
$\delta$-net argument to show that. We assume that $\mathcal{N}$ is a $\delta$-net of the unit
complex sphere in ${\mathbb S}_{\C}^{d-1}$  and hence the covering number $\# \mathcal{N}\le
(1+\frac{2}{\delta})^{2d}$. For any $\vz, \vv\in  \mathbb{S}_{\C}^{d-1}$, there exists a
$\vz_0,\vv_0  \in \mathcal{N}\times \mathcal{N}$ such that $\norm{\vz-\vz_0}\le \delta$ and $\norm{\vv-\vv_0}\le \delta$. Noting $f(\tau):=\tau^2
\phi(\tau/\beta)$ is a bounded function with Lipschitz constant $O(\beta)$, we
obtain that
\begin{equation}\label{eq:gudvz1}
\begin{aligned}
&\Big| \frac 1m \sum_{j=1}^m   b_{j,\Re} (\va_j^* \vz)_{\Re} |\va_j^* \vv|^2 \phi\xkh{\frac {|\va_j^* \vv|} {\beta}} -\frac 1m \sum_{j=1}^m  b_{j,\Re} (\va_j^* \vz_0)_{\Re} |\va_j^* \vv_0|^2 \phi\xkh{\frac {|\va_j^* \vv_0|} {\beta}}  \Big| \\
& \le  \frac 1m \sum_{j=1}^m |b_{j,\Re}| |\va_j^* \vv|^2 \phi\xkh{\frac {|\va_j^* \vv|} {\beta}}  \abs{\va_j^* \vz- \va_j^* \vz_0} \\
& \quad +  \frac 1m \sum_{j=1}^m |b_{j,\Re}| |\va_j^* \vz_0|~  \Big | |\va_j^* \vv|^2 \phi\xkh{\frac {|\va_j^* \vv|} {\beta}} - |\va_j^* \vv_0|^2 \phi\xkh{\frac {|\va_j^* \vv_0|} {\beta}} \Big|  \\
&\lesssim   \frac {\beta^2} m \sum_{j=1}^m |b_{j,\Re}| \abs{\va_j^* \vz- \va_j^* \vz_0} +  \frac {\beta} m \sum_{j=1}^m  |b_{j,\Re}|  |\va_j^* \vz_0|    \abs{\va_j^* \vv- \va_j^* \vv_0} \nonumber\\
& \le \frac {\beta^2} m \norm{\vb} \sqrt{ \sum_{j=1}^m \Abs{\va_j^*(\vz-\vz_0)}^2 } + \frac{\beta}m  \sqrt{\sum_{j=1}^m  |b_j| \Abs{\va_j^*\vz_0}^2 } \cdot  \sqrt{\sum_{j=1}^m  |b_j| \Abs{\va_j^*(\vv-\vv_0)}^2 }  \nonumber\\
&\lesssim  \frac{\beta^2\norm{\vb}}{\sqrt{m}} \norm{\vz-\vz_0} + \beta \xkh{\frac{ \norm{\vb} }{\sqrt m}+1}  \norm{\vv-\vv_0} \le 2 \xkh{ \frac{ \norm{\vb}}{\sqrt{m}} +1}\beta^2\delta  ,
\end{aligned}
\end{equation}
where the fourth inequality follows from Lemma \ref{le:etalowbound} which says, with probability at least $1-2\exp(-c_1 d)$ for a universal constant $c_1>0$, the following  holds
 \begin{eqnarray*}
  \sum_{j=1}^m  |b_j|  |\va_j^* \vw|^2  =   \vw^* \xkh{ \sum_{j=1}^m  |b_j|  \va_j\va_j^*}\vw
&\lesssim&    \xkh{\sum_{j=1}^m  |b_j| + \sqrt{d}\norm{\vb}+d\norms{\vb}_{\infty}}\norm{\vw}^2 \\&\lesssim& \xkh{\sqrt{m} \norm{\vb}+m}\norm{\vw}^2 \quad \mbox{for all} \quad \vw\in \C^d,
 \end{eqnarray*}
provided $ \sum_{j=1}^m  |b_j| \lesssim  m$, $\norms{\vb}_{\infty}  \lesssim \sqrt{\log m}$ and $m\gtrsim d \log m $. Here, we use the fact that
\[
\sum_{j=1}^m  |b_j| \le \sqrt m \sqrt{\sum_{j=1}^m  |b_j|^2 } \le \sqrt m \sqrt[4]{m \sum_{j=1}^m  |b_j|^4}  \lesssim m.
\]
  Choosing $\delta=c_2 \epsilon/\beta^2 $ for some universal constant $c_2>0$ and taking the union bound,  we obtain that
\[
T_1=\abs{ \frac 1m \sum_{j=1}^m  b_{j,\Re} (\va_j^* \vz)_{\Re} |\va_j^* \vv|^2 \phi\xkh{\frac {|\va_j^* \vv|} {\beta}} } \le \xkh{\frac { \norm{\vb}}{ \sqrt{m}}+1}\cdot \frac{\epsilon}2  \quad \mbox{for all} \quad \vz,\vv \in \mathbb{S}_{\C}^{d-1}
\]
holds with probability at least
\[
1-2(1+\frac{2}{\delta})^{2d} \cdot \exp(-c  \epsilon^2 \beta^{-2} m) -2 \exp(-c_1 d) \ge 1-2\exp(-c_3  \epsilon^2 \beta^{-2} m) -2 \exp(-c_1 d)
\]
 provided  $m\ge C\cdot  (\beta/\epsilon)^{2} \log (\beta/\epsilon)  d \log m $. Here,  $C$ and $c_3$ are positive universal constants.  Using the similar argument as above, we  obtain that the following holds with probability at least $1-2\exp(-c_3  \epsilon^2 \beta^{-2} m) -2 \exp(-c_1 d) $:
\[
T_2=  \abs{\frac 1m \sum_{j=1}^m  b_{j,\Im} (\va_j^* \vz)_{\Im} |\va_j^* \vv|^2 \phi\xkh{\frac {|\va_j^* \vv|} {\beta}}} \le  \xkh{\frac { \norm{\vb}}{ \sqrt{m}}+1}\cdot \frac{\epsilon}2 \quad \mbox{for all} \quad \vz,\vv \in \mathbb{S}_{\C}^{d-1}
\]
provided  $m\ge C \cdot (\beta/\epsilon)^{2} \log (\beta/\epsilon)  d\log m $.
%If we take $\beta$ sufficiently large (depending on
%$\epsilon$) then we have
%\begin{equation*}
%\frac 1m \sum_{j=1}^m |\va_j^* \vz|^2 \Re(\vx_0^* \va_j \va_j^* \vz) \phi\xkh{\frac {|\va_j^* \vz|} {\beta}} \le 2\Re(\vx_0^* \vz) + \epsilon  \quad \mbox{for all} \quad \vz\in \mathbb{S}_{\C}^{d-1}
%\end{equation*}
%with probability at least $1-c_a m^{-1}-c'_b \exp(-c'(\beta) \epsilon^2 m /\log m)$,
%provided $m\ge c(\epsilon) d \log d$   for a positive constant $c(\epsilon)$
%depending on $\epsilon$.

Finally, we turn to prove the claim (\ref{eq:claim1}).  We use Cauchy-Schwarz inequality to obtain that
\begin{equation}\label{eq:1claim1}
\frac 1m \sum_{j=1}^m |b_j | |\va_j^* \vz| |\va_j^* \vv|^2\1_{\dkh{|\va_j^* \vv | \ge \beta}} \le \sqrt{\frac 1m \sum_{j=1}^m |\va_j^* \vz|^2 |\va_j^* \vv|^2}\sqrt{\frac 1m \sum_{j=1}^m |b_j|^2 |\va_j^* \vv|^2\1_{|\va_j^* \vv| \ge \beta}}.
\end{equation}
Recall that $ \sum_{j=1}^m  |b_j|^4 \lesssim  m$ and $\|\vb\|_\infty \lesssim \sqrt{m}$.
  According to Lemma \ref{le:tail}, we obtain  that if $m\ge C_1 \epsilon^{-4}\log( 1/\epsilon^{2}) d \log m$ then,  with probability at least $1-2\exp(-c_4 \epsilon^4 m/\log m)-2\exp(-c_5 d) $ ,  the following holds
\begin{equation}\label{eq:2claim1}
\begin{aligned}
 \frac 1m \sum_{j=1}^m |b_j|^2 |\va_j^* \vv|^2\1_{|\va_j^* \vv| \ge \beta} & \le   \xkh{2\beta e^{-0.49 \beta^2}+\frac{\epsilon^2} 2}  \cdot \frac{\norm{\vb}^2}{m}  +\frac{\epsilon^2} 2 \\
 &\le  \xkh{\frac{ \norm{\vb}^2}{m}+1} \cdot \epsilon^2\le   \xkh{ \frac{\norm{\vb}}{\sqrt m}+1} \epsilon
 \end{aligned}
\end{equation}
 for all $\vv \in \mathbb{S}_{\C}^{d-1}$, where $C_1, c_4$ and $c_5$ are positive universal constants. Here, in the second inequality we take $\beta$ to be sufficiently large (depending only on $\epsilon$).
 Combining  (\ref{eq:1claim1}) and (\ref{eq:2claim1}), we arrive at  (\ref{eq:claim1}).

% By Lemma
%\ref{le:indi}, for any $\epsilon>0$ there exists a sufficiently large $\beta$ such
%that  if $m\ge c_1(\epsilon) d$ then with probability at least $1-C\exp(-c_0(
%\epsilon) m)- c_2m^{-1}$ it holds
%\begin{equation}\label{eq:rest}
%\begin{aligned}
%r &\le  \frac 1m \sum_{j=1}^m \abs{\va_j^* \vz}^3  \abs{\va_j^* \vx_0} \1_{\dkh{|\va_j^* \vz| \ge \beta}} \notag\\
%&\le  \xkh{\frac 1m \sum_{j=1}^m \abs{\va_j^* \vz}^4}^{\frac 34} \xkh{\frac 1m \sum_{j=1}^m \abs{\va_j^* \vx_0}^8}^{\frac 18}  \xkh{\frac 1m \sum_{j=1}^m \1_{\dkh{|\va_j^* \vz| \ge \beta}}}^{\frac 18}  \notag \\
%&\le  2 \epsilon \xkh{\frac 1m \sum_{j=1}^m \abs{\va_j^* \vz}^4}^{\frac 34} ,
%\end{aligned}
%\end{equation}
% where we use the Chebyshev's inequality in the last line to deduce that with probability at least $1-c_2m^{-1}$,
%\[
%\frac 1m \sum_{j=1}^m \abs{\va_j(1)}^8 \le 25.
%\]
%Here, $c_0(\epsilon)$ and $c_1(\epsilon)$ are constants depending on $\epsilon$,  and
%$C, c_2$ are absolute constants.
\end{proof}

\begin{lemma} \label{le:axav2}
Suppose that $\va_j \in \C^d, j=1,\ldots,m $, are i.i.d. complex Gaussian random vectors.
Assume that $\vb \in \C^m$ satisfies    $ \sum_{j=1}^m |b_j|^4 \lesssim m$ and $\norms{\vb}_\infty \lesssim \sqrt{\log m}$. Assume that $\vx\in \C^d$ is a fixed vector.
For any $\epsilon \in (0,1)$, if $m\ge C_1 \epsilon^{-2} d \log m$ then,
with probability at least $1-2\exp(-c' \epsilon^2 m/\log m) - c'' \epsilon^{-2}m^{-1}$, it holds that
\[
\left\| \frac 1m \sum_{j=1}^m \zkh{\bar{b}_j (\va_j^* \vx)}_{\Re}  \va_j \va_j^* \right\|_2 \le  \epsilon\xkh{\frac{ \norm{\vb}}{\sqrt m}+ 1} \norm{\vx}.
\]
Here, $C_1, c'$ and $c''$ are positive universal constants.
\end{lemma}

\begin{proof}
We assume that $\mathcal{N}$
 is a $1/4$-net  of the complex unit sphere ${\mathbb S}_{\C}^{d-1}$ with the cardinality $\# \mathcal{N} \le 9^{2d}$. According to  Lemma \ref{le:supnorm}, we have
\begin{equation*}
\left\| \frac 1m \sum_{j=1}^m \zkh{\bar{b}_j (\va_j^* \vx)}_{\Re}  \va_j \va_j^* \right\|_2   \le 2 \max_{\vv \in \mathcal{N}} \Big| \frac 1m \sum_{j=1}^m \zkh{\bar{b}_j (\va_j^* \vx)}_{\Re} |\va_j^* \vv|^2 \Big|.
\end{equation*}

Without loss of generality, we assume $\vx=\ve_1$. Then
\begin{equation}\label{eq:liangterm}
\begin{aligned}
\frac 1m \sum_{j=1}^m \zkh{\bar{b}_j (\va_j^* \vx)}_{\Re} |\va_j^* \vv|^2& = \frac 1m  \abs{v_1}^2\sum_{j=1}^m \xkh{ \bar{b}_j \bar{a}_{j,1}}_{\Re} |\bar{a}_{j,1}  |^2+ \frac 1m \sum_{j=1}^m \xkh{ \bar{b}_j \bar{a}_{j,1}}_{\Re} |\tilde{\va}_{j}^* \tilde{\vv}|^2
\end{aligned}
\end{equation}
where $\tilde{\va}_{j}, \tilde{\vv} \in \C^{d-1}$ are generated by deleting the first entry of  the vector $\va_j$ and $\vv$, respectively.  For the first term, since $\va_j$ are standard  complex Gaussian random variables,  a simple calculation shows that
 \[
 \mbox{Var}\xkh{ \xkh{ \bar{b}_j \bar{a}_{j,1}}_{\Re} |\bar{a}_{j,1} |^2} \le 3\abs{\vb_j}^2.
 \]
%Note that
%$ \frac 1m \sum_{j=1}^m \Re\xkh{ \bar{b}_j \bar{a}_{j,1}} |\bar{a}_{j,1} |^2$ is a centered random variable.
It then follows from Chebyshev's inequality (Lemma \ref{chebv}) that
\begin{equation}\label{eq:chebyle}
\PP\xkh{\abs{ \frac 1m \sum_{j=1}^m \xkh{ \bar{b}_j \bar{a}_{j,1}}_{\Re} |\bar{a}_{j,1} |^2}\ge t } \le  \frac{3\norm{\vb}^2}{m^2 t^2}.
\end{equation}
For any $0<\epsilon<1$, taking $t:=\epsilon  \norm{\vb}/\sqrt{m}$ in (\ref{eq:chebyle}), we obtain that the following holds
\begin{equation} \label{eq:firsv1}
\abs{ \frac 1m \sum_{j=1}^m \xkh{ \bar{b}_j \bar{a}_{j,1}}_{\Re} |\bar{a}_{j,1} |^2}\le \frac{\epsilon \norm{\vb}}{\sqrt m}
\end{equation}
with probability at least $1-3\epsilon^{-2}m^{-1}$.
%provided $\norm{\vb} \lesssim \sqrt{m}$. \textcolor{red}{xu: please explain it in detail}
We next turn to the term  $\frac 1m \sum_{j=1}^m \xkh{ \bar{b}_j \bar{a}_{j,1}}_{\Re} |\tilde{\va}_{j}^* \tilde{\vv}|^2$ in (\ref{eq:liangterm}).
We note that
\[
\frac 1m \sum_{j=1}^m \xkh{ \bar{b}_j \bar{a}_{j,1}}_{\Re} |\tilde{\va}_{j}^* \tilde{\vv}|^2=\frac 1m \sum_{j=1}^m  b_{j,\Re} a_{j,1,\Re} |\tilde{\va}_{j}^* \tilde{\vv}|^2- \frac 1m \sum_{j=1}^m  b_{j,\Im} a_{j,1,\Im} |\tilde{\va}_{j}^* \tilde{\vv}|^2.
\]
For any fixed $\tilde{\vv}_0\in \C^{d-1}$, the terms $ |\tilde{\va}_{j}^* \tilde{\vv}_0|^2- \norm{ \tilde{\vv}_0}^2$ are centered subexponential random variables with the maximal subexponential norm $K:=C \norm{ \tilde{\vv}_0}^2$ where $C>0$ is a universal constant. Furthermore, $|\tilde{\va}_{j}^* \tilde{\vv}_0|^2$ are independent with $a_{j,1}$. We use Bernstein's inequality to obtain that
\begin{equation*}
\begin{aligned}
& \PP\xkh{\frac 1m \sum_{j=1}^m  b_{j,\Re} a_{j,1,\Re} \Big( |\tilde{\va}_{j}^* \tilde{\vv}_0|^2- \norm{ \tilde{\vv}_0}^2 \Big) \ge t \norm{\tilde{\vv}_0}^2  } \\
&\quad\le  2\exp\xkh{-c \min \xkh{\frac{ m^2 t^2}{\sum_{j=1}^m  b_{j,\Re}^2 a_{j,1,\Re}^2 }, \frac{ m t}{ \max_{j\in [m]} \abs{b_{j,\Re}}\cdot \abs{ a_{j,1,\Re} } }  }},
\end{aligned}
\end{equation*}
where $c>0$ is a universal constant.
Taking $t:=\epsilon/2$, together with the union bound over $\mathcal N$, we obtain that
\begin{equation}\label{eq:ch}
\begin{aligned}
& \PP\xkh{\frac 1m \sum_{j=1}^m  b_{j,\Re} a_{j,1,\Re} \Big( |\tilde{\va}_{j}^* \tilde{\vv}|^2- \norm{ \tilde{\vv}}^2 \Big) \ge \frac{\epsilon}2 \norm{\tilde{\vv}}^2  }   \\
&\le  2\exp\xkh{-4c \min \xkh{\frac{ m^2 \epsilon^2}{\sum_{j=1}^m  b_{j,\Re}^2 a_{j,1,\Re}^2 }, \frac{ m \epsilon}{ \max_{j\in [m]} \abs{b_{j,\Re} } \cdot\abs{a_{j,1,\Re}}  }  } + 5d}
\end{aligned}
\end{equation}
for all $\vv \in \mathcal N$.
By Chebyshev's inequality, we obtain that the following holds with probability at least $1-\epsilon^{-2}m^{-1}$
\begin{equation} \label{eq:ch1}
\abs{ \frac 1m \sum_{j=1}^m  b_{j,\Re} a_{j,1,\Re} } \norm{ \tilde{\vv}}^2 \le \frac{\epsilon \norm{\vb_{\Re}}}{\sqrt m} \norm{ \tilde{\vv}}^2.
\end{equation}
Similarly, the Chebyshev's inequality implies
\begin{equation}\label{eq:ch2}
\PP\xkh{\abs{\sum_{j=1}^m  b_{j,\Re}^2 a_{j,1,\Re}^2 - \frac 12 \norm{\vb_{\Re}}^2} \ge m} \le \frac{ \sum_{j=1}^m |b_j|^4}{m^2} \lesssim \frac 1m
\end{equation}
provided $ \sum_{j=1}^m |b_j|^4 \lesssim m$. Moreover, a union bound gives
\begin{equation} \label{eq:ch3}
   \max_{1\le j \le m} |a_{j,1}| \le \sqrt{10 \log m}
\end{equation}
 with probability at least $1- m^{-2}$.
 Putting \eqref{eq:ch1}, \eqref{eq:ch2} and \eqref{eq:ch3} into \eqref{eq:ch}, we obtain that, with probability at least $1-2\exp(-c_1 \epsilon^2 m/\log m) -  c_2 \epsilon^{-2}m^{-1} -m^{-2}$, it holds
 \begin{equation} \label{eq:sevtild}
\abs{ \frac 1m \sum_{j=1}^m  b_{j,\Re} a_{j,1,\Re} |\tilde{\va}_{j}^* \tilde{\vv}|^2} \le \xkh{\frac{ \norm{\vb_{\Re}}}{\sqrt m}+ \frac12} \epsilon \norm{ \tilde{\vv}}^2 \quad \mbox{for all} \quad \vv\in \mathcal N
 \end{equation}
provided $\norm{\vb}\le \sqrt[4]{m\sum_{j=1}^m  |b_j|^4} \lesssim \sqrt m$, $\norms{\vb}_\infty \lesssim \sqrt{\log m}$ and $m\ge C_1 \epsilon^{-2} d \log m$. Here, $C_1, c_1$ and $c_2$ are positive universal constants.
Using the same argument as above, we obtain that
\begin{equation} \label{eq:sevtild2}
\abs{ \frac 1m \sum_{j=1}^m  b_{j,\Im} a_{j,1,\Im} |\tilde{\va}_{j}^* \tilde{\vv}|^2} \le \xkh{\frac{ \norm{\vb_{\Im}}}{\sqrt m}+ \frac12}\epsilon \norm{ \tilde{\vv}}^2 \quad \mbox{for all} \quad \vv\in \mathcal N.
 \end{equation}
Combining  \eqref{eq:firsv1}, \eqref{eq:sevtild} and \eqref{eq:sevtild2}, we obtain that
\[
\frac 1m \sum_{j=1}^m \zkh{\bar{b}_j (\va_j^* \vx)}_{\Re} |\va_j^* \vv|^2 \le \xkh{\frac{ \norm{\vb}}{\sqrt m}+ 1}\epsilon \norm{ \vv}^2  \quad \mbox{for all} \quad \vv\in {\mathcal N}
\]
holds with probability at least $1-2\exp(-c' \epsilon^2 m/\log m) - c'' \epsilon^{-2}m^{-1}$, provided $m\ge C_1 \epsilon^{-2} d \log m$. Here, $c'>0$ and $c''>0$ are universal constants.
 This completes the proof.

\end{proof}

\begin{lemma} \label{le:rebax}
Suppose that $\va_j \in \C^d,j=1,\ldots,m $, are i.i.d. complex Gaussian random vectors.
Assume that $\vb\in \C^m$  which satisfies  $ \sum_{j=1}^m |b_j|^4 \lesssim m$ and $\norms{\vb}_\infty \le \sqrt{\log m}$.
For any $\epsilon \in (0,1)$, if $m\ge C'' \epsilon^{-2}\log (1/\epsilon) ~d \log m$ then,
with probability at least $1-6\exp(-c' \epsilon^2 m/\log m)-6\exp(-c'' d) $,  it holds that
\[
\frac 1m \sum_{j=1}^m \zkh{\bar{b}_j (\va_j^* \vv)}_{\Re}^2  \ge \frac{1-\epsilon}{2} \cdot  \frac{\norm{\vb}^2}{m} -\epsilon \quad \mbox{for all} \; \vv \in \mathbb{S}_{\C}^{d-1}.
\]
Here, $C'',c', c'' >0$ are universal constants.
\end{lemma}

\begin{proof}
A simple calculation shows that
\begin{equation} \label{eq:split}
\zkh{\bar{b}_j (\va_j^* \vv)}_{\Re}^2=b_{j,\Re}^2 (\va_j^* \vv)_{\Re}^2+b_{j,\Im}^2 (\va_j^* \vv)_{\Im}^2+2 b_{j,\Re}b_{j,\Im} (\va_j^* \vv)_{\Re} (\va_j^* \vv)_{\Im}.
\end{equation}
We first give a lower bound for the first term $\frac 1m \sum_{j=1}^m  b_{j,\Re}^2 (\va_j^* \vv)_{\Re}^2$. Note that $\E (\va_j^* \vv)_{\Re}^2=1/2$.  For any fixed $\vv_0 \in \mathbb{S}_{\C}^{d-1}$, by Bernstein's inequality, we have
\[
\PP\xkh{\abs{\frac 1m \sum_{j=1}^m  b_{j,\Re}^2 (\va_j^* \vv_0)_{\Re}^2-\frac{\norm{\vb_{\Re}}^2}{2m}} \ge \frac{\epsilon}6} \le 2\exp\xkh{-c \min\xkh{\frac{\epsilon^2 m^2}{\sum_{j=1}^m b_{j,\Re}^4},\frac{\epsilon m}{\norms{\vb}_{\infty}^2}}},
\]
where $c>0$ is a universal constant.
 Recall that  $ \sum_{j=1}^m |b_j|^4 \lesssim m$ and $\norms{\vb}_{\infty} \lesssim \sqrt{\log m}$.  We  obtain that, with probability at least $1-2\exp(-c_1 \epsilon^2 m/\log m)$, it holds that
\begin{equation} \label{eq:brv2}
\frac 1m \sum_{j=1}^m  b_{j,\Re}^2 (\va_j^* \vv_0)_{\Re}^2 \ge \frac{\norm{\vb_{\Re}}^2}{2m} -\frac{\epsilon}{6},
\end{equation}
where $c_1>0$ is a universal constant.
We next give a uniform bound for \eqref{eq:brv2}.
Suppose that $\mathcal{N}$ is an $\varepsilon_0$-net  over $\mathbb{S}_{\C}^{d-1}$
 with the cardinality $\# \mathcal{N} \le (1+\frac{2}{\varepsilon_0})^{2d}$.
Then for any $\vv \in \mathbb{S}_{\C}^{d-1}$, there exists a $\vv_0 \in \mathcal{N}$ such that $\norm{\vv-\vv_0}\le \varepsilon_0$. Thus, when $m\ge C' d\log m$ for a universal constant $C'>0$,  with probability at least $1-2\exp(- c_2 d)$, it holds that
\begin{equation}\label{eq:2cizheng}
\begin{aligned}
& \abs{ \frac 1m \sum_{j=1}^m  b_{j,\Re}^2 (\va_j^* \vv)_{\Re}^2  -\frac 1m \sum_{j=1}^m  b_{j,\Re}^2 (\va_j^* \vv_0)_{\Re}^2} \\
&\le   \frac{1}{m} \sum_{j=1}^m   b_{j,\Re}^2  |\va_j^* \vv| \abs{ \va_j^* (\vv-\vv_0)}+\frac{1}{m} \sum_{j=1}^m   b_{j,\Re}^2 |\va_j^* \vv_0| \abs{ \va_j^* (\vv-\vv_0)} \\
&\le  \xkh{\sqrt{\frac{1}{m} \sum_{j=1}^m   b_{j,\Re}^2  |\va_j^* \vv|^2}+\sqrt{\frac{1}{m} \sum_{j=1}^m  b_{j,\Re}^2 |\va_j^* \vv_0|^2}} \cdot \sqrt{ \frac{1}{m} \sum_{j=1}^m   b_{j,\Re}^2 \abs{ \va_j^* (\vv-\vv_0)}^2} \\
&\le  2 \xkh{\frac{\norm{\vb_{\Re}}^2}{m} + \sqrt{\frac 1 m \sum_{j=1}^m  |b_j|^4 }+\frac{d\norms{\vb_{\Re}}_{\infty}^2}m } \varepsilon_0 \\
&\lesssim \xkh{\frac{\norm{\vb_{\Re}}^2}{m} + 1}\varepsilon_0
\end{aligned}
\end{equation}
where the third inequality follows from Lemma \ref{le:etalowbound} which says that
 \[
 \norm{ \sum_{j=1}^m   b_{j,\Re}^2  \va_j\va_j^*} \lesssim \norm{\vb_{\Re}}^2+ \sqrt d \sqrt{  \sum_{j=1}^m  |b_j|^4} +d\norms{\vb_{\Re}}_{\infty}^2
 \]
 holds with probability at least $1-2\exp(-c_2 d)$. Here, $c_2>0$ is a universal constant.
 Choosing $\varepsilon_0:= c_3\epsilon$ in (\ref{eq:2cizheng}) for some universal constant $c_3>0$ and taking the union bound over $\mathcal N$, we obtain that
\begin{equation} \label{eq:brajv3}
\frac 1m \sum_{j=1}^m  b_{j,\Re}^2 (\va_j^* \vv)_{\Re}^2  \ge \frac{3-\epsilon}{6} \cdot \frac{\norm{\vb_{\Re}}^2}{m}   -\frac{\epsilon}{3}
\end{equation}
holds for all $\vv \in \mathbb{S}_{\C}^{d-1}$ with probability at least
\[
1-2(1+\frac{2}{\varepsilon_0})^{2d} \cdot  \exp(-c_1 \epsilon^2 m/\log m)-2\exp(-c_2 d) \ge 1-2\exp(-c' \epsilon^2 m/\log m)-2\exp(-c_2 d)
\]
provided $m\ge C'' \epsilon^{-2}\log (1/\epsilon) ~d \log m$, where $C''$ and $c'$ are universal positive constants.

Similarly, for the second and third terms of \eqref{eq:split}, when $m\ge C'' \epsilon^{-2}\log (1/\epsilon) ~d \log m$, with probability at least $1-2\exp(-c' \epsilon^2 m/\log m)-2\exp(-c_2 d) $, the followings hold
\begin{equation} \label{eq:brajv1}
\frac 1m \sum_{j=1}^m  b_{j,\Im}^2 (\va_j^* \vv)_{\Im}^2 \ge  \frac{3-\epsilon}{6} \cdot  \frac{\norm{\vb_{\Im}}^2}{m} -\frac{\epsilon}{3}
\end{equation}
and
\begin{equation}\label{eq:brajv2}
\abs{\frac 1m \sum_{j=1}^m  b_{j,\Re}b_{j,\Im} (\va_j^* \vv)_{\Re} (\va_j^* \vv)_{\Im}} \le \frac{\epsilon}3 \cdot \xkh{\frac{\norm{\vb}^2}{m} +1 }.
\end{equation}
Combining (\ref{eq:brajv1}),  (\ref{eq:brajv2}) and  (\ref{eq:brajv3}),  we  arrive at the conclusion.

%Finally, for the third term of \eqref{eq:split}, it follows from Bernstein's inequality that for any fixed $\vv \in \mathbb{S}_{\C}^{d-1}$, we have
%\[
%\PP\xkh{\abs{\frac 1m \sum_{j=1}^m  b_{j,\Re}b_{j,\Im} (\va_j^* \vv)_{\Re} (\va_j^* \vv)_{\Im} } \ge \frac{\epsilon}6} \le 2\exp\xkh{-c \min\xkh{\frac{\epsilon^2 m^2}{\sum_{j=1}^m b_{j,\Re}^2b_{j,\Im}^2},\frac{\epsilon m}{\norms{\vb}_{\infty}^2}}}.
%\]
%It means  with probability at least $1-2\exp(-c\epsilon^2 m/\log m)$
%\[
%\abs{\frac 1m \sum_{j=1}^m  b_{j,\Re}b_{j,\Im} (\va_j^* \vv)_{\Re} (\va_j^* \vv)_{\Im}} \le \frac{\epsilon}3.
%\]

%obtain the conclusion that when  $m\ge C'' \epsilon^{-2}\log (1/\epsilon) ~d \log m$, with probability at least $1-6\exp(-c' \epsilon^2 m/\log m)-6\exp(-c_2 d) $ it holds
%\[
%\frac 1m \sum_{j=1}^m \zkh{\bar{b}_j (\va_j^* \vv)}_{\Re}^2  \ge \frac{1-\epsilon}{2} \cdot  \frac{\norm{\vb}^2}{m} -\epsilon.
%\]
\end{proof}

\begin{lemma} \label{le:rbavrazav}
Suppose that $\va_j \in \C^d,j=1,\ldots,m $, are i.i.d. complex Gaussian random vectors.
Assume that $\vb \in \C^m$ is a vector obeying $ \sum_{j=1}^m  |b_j|^4 \lesssim m$ and $\norms{\vb}_{\infty} \le \sqrt{\log m}$.
For any $\epsilon \in (0,1)$, if $m\ge  c(\epsilon) d \log m $ then the following holds
with probability at least $1-6 \exp(-c'(\epsilon) m/\log m ) - 6\exp(-c'' d)$:
\[
\abs{\frac 1m \sum_{j=1}^m \zkh{\bar{b}_j(\va_j^* \vv) }_{\Re} \zkh{(\va_j^* \vz) (\va_j^* \vv)}_{\Re} }\le  \epsilon\xkh{ \frac{\norm{\vb}}{\sqrt m}+1} \xkh{1 + \xkh{\frac 1m \sum_{j=1}^m |\va_j^* \vz|^2 |\va_j^* \vv|^2}^{\frac 12}  }
\]
for all $ \vz, \vv \in \mathbb{S}_{\C}^{d-1}$, where $c''$ is a positive universal constant and $c(\epsilon)$, $c'(\epsilon)$ are positive constants
 depending only on $\epsilon$.
\end{lemma}
\begin{proof}
Suppose that $\phi\in C_c^{\infty}(\mathbb R)$ is a Lipschitz continuous function
satisfying $0\le \phi(x)\le 1$ for all $x\in \mathbb R$. We furthermore require
$\phi(x)=1$ for $|x|\le 1$ and $\phi(x)=0$ for $|x|\ge 2$. For any   $\beta>0$, we
have
\begin{equation}\label{eq:T12r}
\abs{\frac 1m \sum_{j=1}^m \zkh{\bar{b}_j(\va_j^* \vv) }_{\Re} \zkh{(\va_j^* \vz) (\va_j^* \vv)}_{\Re} }\le  T_1+T_2 +r,
% &=& \abs{\frac 1m \sum_{j=1}^m \Re(\bar{b}_j(\va_j^* \vv) ) \Re\xkh{(\va_j^* \vz) (\va_j^* \vv)}  \phi^2\xkh{\frac {|\va_j^* \vv|} {\beta}} } \\
% && + \abs{\frac 1m \sum_{j=1}^m \Re(\bar{b}_j(\va_j^* \vv) ) \Re\xkh{(\va_j^* \vz) (\va_j^* \vv)} \xkh{1-\phi^2\xkh{\frac {|\va_j^* \vv|} {\beta}} }} \nonumber\\
\end{equation}
where
\begin{equation*}
\begin{aligned}
T_1&:= \abs{\frac 1m \sum_{j=1}^m b_{j,\Re} (\va_j^* \vv)_{\Re}  \zkh{(\va_j^* \vz) (\va_j^* \vv)}_{\Re}  \phi^2 \xkh{\frac {|\va_j^* \vv|} {\beta}}} , \\
 T_2&:=\abs{ \frac 1m \sum_{j=1}^m b_{j,\Im} (\va_j^* \vv)_{\Im} \zkh{(\va_j^* \vz) (\va_j^* \vv)}_{\Re} \phi^2\xkh{\frac {|\va_j^* \vv|} {\beta}}} ,\\
  r&:=\frac 1m \sum_{j=1}^m |b_j| |\va_j^* \vz| |\va_j^* \vv|^2\1_{\dkh{|\va_j^* \vv | \ge \beta}}  .
  \end{aligned}
\end{equation*}
%where
%\begin{equation*}
%\begin{aligned}
% T&:= \frac 1m \sum_{j=1}^m |\va_j^* \vz|^2 \Re(\vx_0^* \va_j \va_j^* \vz) \phi\xkh{\frac {|\va_j^* \vz|} {\beta}},\\
%r&:=\frac 1m \sum_{j=1}^m |\va_j^* \vz|^2 \Re(\vx_0^* \va_j \va_j^* \vz)
% \xkh{1- \phi\xkh{\frac {|\va_j^* \vz|}{\beta}}} \le \frac 1m \sum_{j=1}^m \abs{\va_j^* \vz}^3  \abs{\va_j^* \vx_0} \1_{\dkh{|\va_j^* \vz| \ge \beta}}.
% \end{aligned}
%\end{equation*}
We claim that,  for any $0<\epsilon<1$, there exists a  sufficiently large constant $\beta>1$
such that   if $m\ge c(\epsilon) d \log m $ then, with probability
at least $1-6 \exp(-c'(\epsilon) m/\log m ) - 6\exp(c'' d)$, the followings hold
\begin{equation}\label{eq:claim00}
T_1\le  \frac{\epsilon}{2} \cdot\xkh{ \frac{\norm{\vb}}{\sqrt m}+1},\quad T_2\le  \frac{\epsilon}{2} \cdot \xkh{ \frac{\norm{\vb}}{\sqrt m}+1}
\end{equation}
and
\begin{equation} \label{eq:claim11}
r \le  \epsilon \xkh{ \frac{\norm{\vb}}{\sqrt m}+1} \xkh{\frac 1m \sum_{j=1}^m |\va_j^* \vz|^2 |\va_j^* \vv|^2}^{\frac 12},
\end{equation}
for all $ \vz, \vv  \in \mathbb{S}_{\C}^{d-1}$.
 Here
 $c(\epsilon), c'(\epsilon)  $ are constants depending only on $\epsilon$ and $c''$ is a positive universal constant. Combining (\ref{eq:T12r}),  (\ref{eq:claim00}) and \eqref{eq:claim11}, we  obtain  that
\[
\abs{\frac 1m \sum_{j=1}^m \zkh{\bar{b}_j(\va_j^* \vv) }_{\Re} \zkh{(\va_j^* \vz) (\va_j^* \vv)}_{\Re} }\le  \epsilon\xkh{ \frac{\norm{\vb}}{\sqrt m}+1} \xkh{1 + \xkh{\frac 1m \sum_{j=1}^m |\va_j^* \vz|^2 |\va_j^* \vv|^2}^{\frac 12}  }
\]
for all $ \vz, \vv \in \mathbb{S}_{\C}^{d-1}$.

It remains to prove the claims (\ref{eq:claim00}) and (\ref{eq:claim11}). For any fixed $\vz_0,\vv_0 \in \mathbb{S}_{\C}^{d-1}$, due to the cut-off $\phi\xkh{\frac {|\va_j^* \vv_0|}{\beta}}$, the terms
$  (\va_j^* \vv_0)_{\Re} \Re\xkh{(\va_j^* \vz_0) (\va_j^* \vv_0)}  \phi^2\xkh{\frac {|\va_j^* \vv_0|} {\beta}}$ are centered, independent sub-gaussian random variables with the sub-gaussian norm
$O(\beta)$. According to Hoeffding's inequality, we obtain that the following holds
with probability at least $1-2\exp(-c  \epsilon^2 \beta^{-2} m)$
\begin{equation}\label{eq:gudvz0}
\abs{\frac 1m \sum_{j=1}^m   b_{j,\Re} (\va_j^* \vv_0)_{\Re} \zkh{(\va_j^* \vz_0) (\va_j^* \vv_0)}_{\Re} \phi^2\xkh{\frac {|\va_j^* \vv_0|} {\beta}} }  \le \frac {\epsilon \norm{\vb}}{4 \sqrt{m}},
\end{equation}
 where $c >0$ is a universal constant.
 We next show that (\ref{eq:gudvz0}) holds
for all unit vectors $\vz, \vv \in \C^d$, for which we adopt a basic version of a
$\delta$-net argument. We assume that $\mathcal{N}$ is a $\delta$-net of the unit
complex sphere in $\C^d$  and hence the covering number $\# \mathcal{N}\le
(1+\frac{2}{\delta})^{2d}$. For any $\vz, \vv\in  \mathbb{S}_{\C}^{d-1}$, there exists a
$\vz_0,\vv_0  \in \mathcal{N}\times \mathcal{N}$ such that $\norm{\vz-\vz_0}\le \delta$ and $\norm{\vv-\vv_0}\le \delta$. Noting $f(\tau):=\tau^2
\phi(\tau/\beta)$ is a bounded function with Lipschitz constant $O(\beta)$, we obtain  that  if $m\gtrsim d\log m$ then, with probability at least $1-2\exp(-c_1 d)$,  it holds that
\begin{equation}\label{eq:gudvz1}
\begin{aligned}
&\Big| \frac 1m \sum_{j=1}^m   b_{j,\Re} (\va_j^* \vv)_{\Re}  \zkh{(\va_j^* \vz) (\va_j^* \vv)}_{\Re}   \phi^2\xkh{\frac {|\va_j^* \vv|} {\beta}} -\frac 1m \sum_{j=1}^m   b_{j,\Re} (\va_j^* \vv_0)_{\Re}  \zkh{(\va_j^* \vz_0) (\va_j^* \vv_0)}_{\Re}  \phi^2\xkh{\frac {|\va_j^* \vv_0|} {\beta}}  \Big| \\
& \le  \frac 1m \sum_{j=1}^m |b_{j,\Re}|  |\va_j^* \vv| |\va_j^* \vz|  \phi\xkh{\frac {|\va_j^* \vv|} {\beta}}  \abs{|\va_j^* \vv| \phi\xkh{\frac {|\va_j^* \vv|} {\beta}} - |\va_j^* \vv_0| \phi\xkh{\frac {|\va_j^* \vv_0|} {\beta}}} \\
&\quad +\frac 1m \sum_{j=1}^m |b_{j,\Re}|  |\va_j^* \vv| \phi\xkh{\frac {|\va_j^* \vv|} {\beta}} |\va_j^* \vv_0| \phi\xkh{\frac {|\va_j^* \vv_0 |} {\beta}}  \abs{\va_j^* \vz- \va_j^* \vz_0} \\
&\quad +  \frac 1m \sum_{j=1}^m |b_{j,\Re}| |\va_j^* \vz_0|  |\va_j^* \vv_0| \phi\xkh{\frac {|\va_j^* \vv_0|} {\beta}}  \Big | |\va_j^* \vv| \phi\xkh{\frac {|\va_j^* \vv|} {\beta}} - |\va_j^* \vv_0| \phi\xkh{\frac {|\va_j^* \vv_0|} {\beta}} \Big|  \\
&\lesssim   \frac {\beta} m \sum_{j=1}^m |b_{j,\Re}| |\va_j^* \vv| \abs{\va_j^* \vv- \va_j^* \vv_0}+ \frac {\beta^2} m \sum_{j=1}^m |b_{j,\Re}|  \abs{\va_j^* \vz- \va_j^* \vz_0}   +  \frac {\beta} m \sum_{j=1}^m  |b_{j,\Re}| |\va_j^* \vz_0|  \abs{\va_j^* \vv- \va_j^* \vv_0} \nonumber\\
&\lesssim  \frac{\beta^2\norm{\vb}}{\sqrt{m}} \norm{\vz-\vz_0} + 2\beta \xkh{\frac{ \norm{\vb} }{\sqrt m}+1}  \norm{\vv-\vv_0} \le 2 \xkh{ \frac{ \norm{\vb}}{\sqrt{m}} +1}\beta^2\delta  ,
\end{aligned}
\end{equation}
provided $ \sum_{j=1}^m  |b_j| \lesssim m$ due to the fact $\sum_{j=1}^m  |b_j|^4 \lesssim m$, where the third inequality follows from Lemma \ref{le:etalowbound}.  Here, $c_1>0$ is a universal constant.  Choosing $\delta=c_2 \epsilon/\beta^2 $ for some universal constant $c_2>0$ and taking the union bound,  we obtain that
\[
T_1=\abs{ \frac 1m \sum_{j=1}^m  b_{j,\Re} (\va_j^* \vz)_{\Re} |\va_j^* \vv|^2 \phi\xkh{\frac {|\va_j^* \vv|} {\beta}} } \le \xkh{\frac { \norm{\vb}}{ \sqrt{m}}+1}\cdot \frac{\epsilon}2  \quad \mbox{for all} \quad \vz,\vv \in \mathbb{S}_{\C}^{d-1}
\]
holds with probability at least
\[
1-2(1+\frac{2}{\delta})^{2d} \cdot \exp(-c  \epsilon^2 \beta^{-2} m) -2 \exp(-c_1 d) \ge 1-2\exp(-c_3  \epsilon^2 \beta^{-2} m) -2 \exp(-c_1 d)
\]
 provided  $m\ge C\cdot  (\beta/\epsilon)^{2} \log (\beta/\epsilon)  d \log m $. Here,  $C$ and $c_3$ are positive universal constants.
Using the method similar to the proof of Lemma \ref{le:azav2}, we can obtain the bounds for  $T_2$ and $r$.
We omit the detail here.
\end{proof}

\section{Proofs of technical results in Section \ref{se:opGD}}
\label{appendix:B}
{\noindent\it \textbf{Proof of Lemma} \ref{le:bound4x}}~
A simple calculation shows that
\[
\frac 1m \sum_{j=1}^m \abs{\va_j^* \vx +b_j}^2=\frac 1m \sum_{j=1}^m \abs{\va_j^* \vx }^2+ \frac 2m \sum_{j=1}^m \zkh{\bar{b}_j (\va_j^* \vx)}_{\Re} + \frac{\norm{\vb}^2}m.
\]
We first consider the term $\frac 1m \sum_{j=1}^m |\va_j^* \vx |^2$. We use Bernstein's inequality to obtain that
\[
(1-\epsilon) \norm{\vx}^2  \le \frac 1m \sum_{j=1}^m \abs{\va_j^* \vx }^2 \le (1+\epsilon) \norm{\vx}^2\quad \text{for any }0<\epsilon\le 1,
\]
holds with probability at least $1-2\exp(-c_1 \epsilon^2 m)$.  Here, $c_1>0$ is a universal constant.
For the second term, noting that
\[
\frac 1m \sum_{j=1}^m \zkh{\bar{b}_j (\va_j^* \vx)}_{\Re}=\frac 1m \sum_{j=1}^m {b}_{j,\Re} (\va_j^* \vx)_{\Re}+\frac 1m \sum_{j=1}^m {b}_{j,\Im} (\va_j^* \vx)_{\Im},
\]
we use Hoeffding's inequality to obtain that
\[
\abs{\frac 1m \sum_{j=1}^m b_{j,\Re} (\va_j^* \vx)_{\Re} }  \le \frac{\epsilon \norm{\vb_{\Re}}}{\sqrt m} \norm{\vx},
\]
and
\[
\abs{\frac 1m \sum_{j=1}^m b_{j,\Im} (\va_j^* \vx)_{\Im} }  \le \frac{\epsilon \norm{\vb_{\Im}}}{\sqrt m} \norm{\vx}
\]
hold with probability at least $1-2\exp(-c_2 \epsilon^2 m)$, where $c_2>0$ is a universal constant.

Recall that $\norm{\vb} \le c_3 \sqrt{m} \norm{\vx}$ for a universal constant $c_3>0$. Collecting the above results, we obtain that, with probability at least $1-4\exp(-c_4 \epsilon^2 m)$, the following  holds
\[
(1- c_0\epsilon)  \norm{\vx}^2 +\frac{\norm{\vb}^2}m\le \frac 1m \sum_{j=1}^m \abs{\va_j^* \vx +b_j}^2 \le (1+c_0\epsilon)  \norm{\vx}^2 +\frac{\norm{\vb}^2}m
\]
for a universal constant $c_0:=2\sqrt2 c_3+1$. Here, $c_4>0$ is a universal constant.
Taking $\epsilon:=\frac3{4c_0}$, we obtain that the following holds  with probability at least $1-4\exp(-c m)$:
\begin{equation*}
\frac14 \norm{\vx}^2 \le \frac 1m\sum_{j=1}^m y_j -\frac{\norm{\vb}^2}{m} \le \frac74 \norm{\vx}^2,
\end{equation*}
where $c>0$ is a universal constant, which implies
\begin{equation*} 
\frac13 R_0 \le \norm{\vx}\leq  R_0,
\end{equation*}
where $R_0:=2\xkh{\frac 1m\sum_{j=1}^m y_j -\frac{\norm{\vb}^2}{m} }^{1/2}$.
 This completes the proof.
\qed

\vspace{2em}
{\noindent\it \textbf{Proof of Lemma} \ref{le:Lipcondition}}~
For any $\vz,\vz' \in \mathcal{S}_R$, we have
\begin{equation} \label{le3.2_1}
\begin{aligned}
 &\norm{\nabla f(\vz) -\nabla f(\vz')} \\
 & =  \frac{\sqrt 2}m \left\| \sum_{j=1}^m  \xkh{|\va_j^* \vz+b_j|^2-y_j}  \xkh{\va_j^* \vz+b_j} \va_j-\sum_{j=1}^m  \xkh{|\va_j^* \vz'+b_j|^2-y_j}  \xkh{\va_j^* \vz'+b_j} \va_j   \right\|_2 \\
 & \le  \sqrt 2 \left\| \frac 1m \sum_{j=1}^m  \xkh{|\va_j^* \vz+b_j|^2- |\va_j^* \vz'+b_j|^2} \va_j\va_j^* \vz' \right\|_2+  \sqrt 2 \left\| \frac 1m \sum_{j=1}^m   |\va_j^* \vz+b_j|^2 \va_j\va_j^* (\vz-\vz') \right\|_2 \\
 & \quad + \sqrt 2 \left\| \frac 1m \sum_{j=1}^m  \xkh{|\va_j^* \vz+b_j|^2- |\va_j^* \vz'+b_j|^2} b_j \va_j \right\|_2+  \sqrt 2 \left\| \frac 1m \sum_{j=1}^m   |\va_j^* \vx+b_j|^2 \va_j\va_j^* (\vz-\vz') \right\|_2 .
 \end{aligned}
\end{equation}
Since $\va_j  \in \C^d$ are complex Gaussian random vectors, with probability at least $1-c_a m^{-d}$, the following  holds
 \[
\max_{1\le j \le m} \norm{\va_j} \le 2 \sqrt{d \log m},
\]
where $c_a>0$ is a universal constant.
Lemma \ref{le:dizwbj} implies that when $m\ge Cd$ for a universal constant $C>0$,
with probability at least $1-3\exp(-cm)$, it holds that
\begin{eqnarray*}
\frac 1m \sum_{j=1}^m \abs{|\va_j^* \vz+b_j|^2-|\va_j^* \vz'+b_j|^2 } \le 3 \xkh{ R+  \frac{ \norm{\vb}}{\sqrt m } } \norm{\vz-\vz'} \quad \mbox{for all} \; \vz,\vz' \in \mathcal{S}_R,
\end{eqnarray*}
where $c>0$ is a universal constant.
  Combining the above two estimators, we obtain that the following holds with probability at least $1-3\exp(-cm)-c_a m^{-d}$:
\begin{equation}\label{le3.2_2}
\begin{aligned}
& \left\| \frac 1m \sum_{j=1}^m  \xkh{|\va_j^* \vz+b_j|^2- |\va_j^* \vz'+b_j|^2} \va_j\va_j^* \vz' \right\|_2 \\
& \le   \max_{1\le j\le m}\norm{\va_j}^2\cdot  \norm{\vz'} \cdot \frac 1m \sum_{j=1}^m \abs{|\va_j^* \vz+b_j|^2-|\va_j^* \vz'+b_j|^2 }  \\
  & \le  12\cdot d\cdot\log m \cdot R\cdot \xkh{R+\frac{ \norm{\vb}}{\sqrt m}}\cdot\norm{\vz-\vz'},
  \end{aligned}
\end{equation}
provided  $m\ge Cd$.
Here, we use the fact that $\norm{\vz'} \le R$ due to $\vz' \in \mathcal{S}_R$. Using the same argument as above, we obtain that
\begin{equation}
\left\| \frac 1m \sum_{j=1}^m  \xkh{|\va_j^* \vz+b_j|^2- |\va_j^* \vz'+b_j|^2} b_j \va_j \right\|_2 \le 6 \norms{\vb}_{\infty} \sqrt{d\log m} \xkh{R+\frac{ \norm{\vb}}{\sqrt m}} \norm{\vz-\vz'}. \label{le3.2_3}
\end{equation}
 Noting that $\norm{\frac 1m \sum_{j=1}^m \va_j\va_j^*} \le 2 $ holds with probability at least $1-\exp(-c m)$, we obtain that
\begin{equation} \label{le3.2_4}
\begin{aligned}
 \left\| \frac 1m \sum_{j=1}^m   |\va_j^* \vz+b_j|^2 \va_j\va_j^* (\vz-\vz') \right\|_2
& \le \Bigg\| \frac 1m \sum_{j=1}^m   |\va_j^* \vz+b_j|^2 \va_j\va_j^* \Bigg\|_2 \norm{\vz-\vz'} \\
&\le  2 \xkh{\max_{1\le j\le m}\norm{\va_j}^2\cdot  \norm{\vz}^2+\norms{\vb}_{\infty}^2} \cdot  \Big\| \frac 1m \sum_{j=1}^m \va_j\va_j^*\Big\| \cdot  \norm{\vz-\vz'} \\
&\le  4\xkh{4R^2d\log m+\norms{\vb}_{\infty}^2 } \norm{\vz-\vz'}.
\end{aligned}
\end{equation}
Here, we use the inequality $ |\va_j^* \vz+b_j|^2 \le 2(|\va_j^* \vz|^2+|b_j|^2)$ for any $j$. Similarly,
\begin{equation}
\left\| \frac 1m \sum_{j=1}^m   |\va_j^* \vx+b_j|^2 \va_j\va_j^* (\vz-\vz') \right\|_2 \le 4\xkh{4\norm{\vx}^2d\log m+\norms{\vb}_{\infty}^2 } \norm{\vz-\vz'}. \label{le3.2_5}
\end{equation}
Substituting \eqref{le3.2_2}, \eqref{le3.2_3}, \eqref{le3.2_4} and \eqref{le3.2_5} into \eqref{le3.2_1}, we obtain that when $m\ge Cd$, with probability at least $1-4\exp(-cm)-c_a m^{-d}$, it holds that
\[
\norm{\nabla f(\vz) -\nabla f(\vz')} \le C_R \norm{\vz-\vz'} \quad \mbox{for all} \quad \vz,\vz' \in \mathcal{S}_R,
\]
where
\[
C_R=6\sqrt 2 \xkh{2R d\log m +\norms{\vb}_{\infty} \sqrt{d\log m} }\xkh{R+\frac{ \norm{\vb}}{\sqrt m}}+8\sqrt 2 \Big( 2d\log m (R^2+\norm{\vx}^2)+ \norms{\vb}_{\infty}^2 \Big).
\]
This completes the proof.
\qed

%\begin{acknowledgements}
%Z. Xu was supported  by Beijing Natural Science Foundation (Z180002)  and by NSFC grant (11688101). 
%M. Huang acknowledges support from Yang Wang and the Department of Mathematics, The Hong Kong University of Science and Technology.
%\end{acknowledgements}

% Authors must disclose all relationships or interests that 
% could have direct or potential influence or impart bias on 
% the work: 
%
% \section*{Conflict of interest}
%
% The authors declare that they have no conflict of interest.

% BibTeX users please use one of
%\bibliographystyle{spbasic}      % basic style, author-year citations
%\bibliographystyle{spmpsci}      % mathematics and physical sciences
%\bibliographystyle{spphys}       % APS-like style for physics
%\bibliography{}   % name your BibTeX data base

% Non-BibTeX users please use

\end{document}